\documentclass[cleveref,thm-restate]{lipics-v2021}
\pdfoutput=1
\nolinenumbers
\hideLIPIcs

\usepackage{xspace}
\usepackage{siunitx}

\NewDocumentCommand\xDeclarePairedDelimiter{mmm}
{%
	\NewDocumentCommand#1{som}{%
		\IfNoValueTF{##2}
		{\IfBooleanTF{##1}{#2##3#3}{\left#2##3\right#3}}
		{\mathopen{##2#2}##3\mathclose{##2#3}}%
	}%
}

\xDeclarePairedDelimiter\abs{\lvert}{\rvert}
\xDeclarePairedDelimiter\norm{\lVert}{\rVert}
\xDeclarePairedDelimiter{\ceil}{\lceil}{\rceil}
\xDeclarePairedDelimiter{\floor}{\lfloor}{\rfloor}
\xDeclarePairedDelimiter{\set}{\lbrace}{\rbrace}

\usepackage{booktabs}
\usepackage{colortbl}
\usepackage{threeparttable}

\newtheorem{reduction}{Reduction}
\crefname{reduction}{reduction}{reductions}

\usepackage{todonotes}

\makeatletter
\g@addto@macro\bfseries{\boldmath}
\makeatother

\usepackage{tikz}
\usetikzlibrary{
  arrows.meta,
  backgrounds,
  fit,
  quotes,
  positioning,
  shapes.geometric,
  backgrounds,
  decorations.pathmorphing,
  decorations.markings,
  calc,
  graphs,
  graphs.standard,
  external,
  patterns
}

\usepackage{shellesc}
\newcommand{\figurecache}{./figure-cache/}
\usetikzlibrary{external}
\ifcase\ShellEscapeStatus
  \PackageWarning{externalizing}{No Shell Escape available, invoke externalization manually}
  \tikzset{external/mode=list and make}
\or
  \ShellEscape{mkdir -p \figurecache}
\else
  \PackageWarning{externalizing}{Shell Escape restricted, make sure ./\figurecache/ exists and invoke `make -j -f \jobname.makefile`}
  \tikzset{external/mode=list and make}
\fi

\tikzsetexternalprefix{\figurecache}

\tikzset{external/only named=true}

\newcommand{\tikzpicturename}[1]{\tikzsetnextfilename{\tikzexternalrealjob-#1}}

\colorlet{Obs}{green!50!teal}

\tikzset{
  node/.style={thick,draw,fill,black,thick,minimum size=\pgfplotmarksize,inner sep=0mm,outer sep=0mm},
  input/.style={
    node,propagating=no,
  },
  output/.style={
    node,rectangle,
  },
  gate/.style={
    node,propagating,draw=black!50,fill=none,inner sep=2.5pt,
  },
  and/.style={
    gate,label={[font=\footnotesize,black]center:$\wedge$}
  },
  or/.style={
    gate,label={[font={\footnotesize},black]center:$\vee$}
  },
  booster/.style={
    node,diamond,
    observed, fill=white, thick,
  },
  >/.tip={Stealth[]},
  snaked/.style={decorate,decoration={snake,post length=1.5mm,amplitude=0.5mm,segment length=2mm}},
  pds instance/.style={
    mark size=1ex,
    label distance=0.5\pgfplotmarksize,
    graphs/radius=0.5cm,
  },
  nonedge/.style={red,decoration={markings,mark=at position 0.5 with {
      \draw (-.6\pgfplotmarksize,-.6\pgfplotmarksize) -- (.6\pgfplotmarksize,.6\pgfplotmarksize);
      \draw (.6\pgfplotmarksize,-.6\pgfplotmarksize) -- (-.6\pgfplotmarksize,.6\pgfplotmarksize);
    }},postaction=decorate
  },
  propagating/.search also={/tikz},
  propagating/yes/.style={circle},
  propagating/no/.style={regular polygon,regular polygon sides=3,minimum size=1.1\pgfplotmarksize},
  propagating/dontcare/.style={rectangle,minimum size=0.9\pgfplotmarksize},
  propagating/.style={propagating/.cd,draw,fill,#1,/tikz},
  propagating/.default=yes,
  observed/.search also={/tikz},
  observed/yes/.style={draw=Obs},
  observed/no/.style={},
  observed/dontcare/.style={fill=black!50!white},
  observed/.style={observed/.cd,draw,#1,/tikz},
  observed/.default=yes,
  solution/selected/.style={red,fill=red!70,observed/yes/.style={}},
  solution/yes/.style={selected},
  solution/excluded/.style={fill=white},
  solution/no/.style={excluded},
  solution/dontcare/.style={fill=black!30},
  solution/undecided/.style={},
  solution/.search also={/tikz},
  solution/.style={solution/.cd,draw,fill,#1,/tikz},
  solution/.default={yes},
}

\usepackage{pgfplots}
\usepackage{pgfplotstable}
\usepgfplotslibrary{groupplots,colorbrewer,statistics,fillbetween}

\pgfplotsset{compat=newest}
\pgfplotsset{cycle list/Dark2,colormap/Dark2}

\pgfplotsset{table/col sep=comma}

\makeatletter
\pgfplotsset{
  boxplot prepared from table/.code={
    \def\tikz@plot@handler{\pgfplotsplothandlerboxplotprepared}%
    \pgfplotsset{
      /pgfplots/boxplot prepared from table/.cd,
      #1,
    }
  },
  /pgfplots/boxplot prepared from table/.cd,
  table/.code={\pgfplotstablecopy{#1}\to\boxplot@datatable},
  row/.initial=0,
  make style readable from table/.style={
    #1/.code={
      \pgfplotstablegetelem{\pgfkeysvalueof{/pgfplots/boxplot prepared from table/row}}{##1}\of\boxplot@datatable
      \pgfplotsset{boxplot/#1/.expand once={\pgfplotsretval}}
    }
  },
  make style readable from table=lower whisker,
  make style readable from table=upper whisker,
  make style readable from table=lower quartile,
  make style readable from table=upper quartile,
  make style readable from table=median,
  make style readable from table=lower notch,
  make style readable from table=upper notch,
}

\colorlet{@initialcolor}{.}
\newcommand{\@definemycolor}[1]{\pgfkeys{/pgfplots/index of colormap={#1}}\colorlet{color#1}{.}}
\@definemycolor{0}
\@definemycolor{1}
\@definemycolor{2}
\@definemycolor{3}
\@definemycolor{4}
\@definemycolor{5}
\@definemycolor{6}
\@definemycolor{7}
\@definemycolor{8}
\@definemycolor{9}
\color{@initialcolor}
\makeatother

\makeatother

\newcommand{\pdfsc}[1]{\texorpdfstring{\textsc{#1}}{#1}}
\newcommand*{\DS}{\pdfsc{Dominating Set}\xspace}
\newcommand*{\Ds}{\pdfsc{DS}\xspace}
\newcommand*{\PDS}{\pdfsc{Power Dominating Set}\xspace}
\newcommand*{\Pds}{\pdfsc{PDS}\xspace}

\newcommand*{\IPDS}{\pdfsc{Implicating \PDS}\xspace}
\newcommand*{\Ipds}{\pdfsc{IPDS}\xspace}
\newcommand*{\HS}{\pdfsc{Hitting Set}\xspace}

\newcommand*{\extension}[1]{#1-\pdfsc{Extension}}
\newcommand*{\Extension}[1]{#1 \pdfsc{Extension}}
\newcommand*{\WCS}{\pdfsc{Weighted Circuit Satisfiability}\xspace}
\newcommand*{\Wcs}{\pdfsc{WCS}\xspace}
\newcommand*{\WMCS}{\pdfsc{Weighted Monotone Circuit Satisfiability}\xspace}
\newcommand*{\Wmcs}{\pdfsc{WMCS}\xspace}

\newcommand*{\W}[1]{{W[#1]}}
\newcommand*{\NP}{\text{NP}}

\newcommand*{\AND}{\texttt{and}}
\newcommand*{\OR}{\texttt{or}}
\newcommand*{\NOT}{\texttt{not}}

\newcommand*{\Vin}{\ensuremath{V_\texttt{in}}}
\newcommand*{\Vand}{\ensuremath{V_\AND}}
\newcommand*{\Vor}{\ensuremath{V_\OR}}

\newcommand*{\TRUE}{\textsc{true}\xspace}
\newcommand*{\FALSE}{\textsc{false}\xspace}

\title{An Efficient Algorithm for Power Dominating Set}
\author{Thomas Bläsius}{Karlsruhe Institute of Technology (KIT), Germany}{thomas.blaesius@kit.edu}{https://orcid.org/0000-0003-2450-744X}{}
\author{Max Göttlicher}{Karlsruhe Institute of Technology (KIT), Germany}{max.goettlicher@kit.edu}{https://orcid.org/0000-0002-5556-4140}{}
\authorrunning{T. Bläsius and M. Göttlicher}
\Copyright{Thomas Bläsius and Max Göttlicher}



\acknowledgements{This work was supported by the German Research Foundation (DFG) as part of the Research Training Group GRK 2153:Energy Status Data -- Informatics Methods for its Collection, Analysis and Exploitation.}

\begin{document}
  \maketitle

  \begin{abstract}
    The problem \PDS (\Pds) is motivated by the placement of phasor measurement units to monitor electrical networks.
    It asks for a minimum set of vertices in a graph that observes all remaining vertices by exhaustively applying two observation rules.
    Our contribution is twofold.
    First, we determine the parameterized complexity of \Pds by proving it is $W[P]$-complete when parameterized with respect to the solution size.
    We note that it was only known to be $W[2]$-hard before.
    Our second and main contribution is a new algorithm for \Pds that efficiently solves practical instances.

    Our algorithm consists of two complementary parts.
    The first is a set of reduction rules for \Pds that can also be used in conjunction with previously existing algorithms.
    The second is an algorithm for solving the remaining kernel based on the implicit hitting set approach.
    Our evaluation on a set of power grid instances from the literature shows that our solver outperforms previous state-of-the-art solvers for \Pds by more than one order of magnitude on average.
    Furthermore, our algorithm can solve previously unsolved instances of continental scale within a few minutes.
  \end{abstract}
  \clearpage
  \section{Introduction}
\label{sec:intro}

Monitoring power voltages and currents in electric grids is vital for maintaining their stability and for cost-effective operation.
The sensors required to obtain high-resolution measurements, so-called phasor measurement units, are expensive pieces of equipment.
The goal to place as few of those sensors as possible to minimize cost is called the \PDS problem (\Pds).
It was first posed by Mili, Baldwin and Adapa.~\cite{mili1990PhasorMP} and formalized by Baldwin et al.~\cite{baldwin1993PowerSO}.
In its basic form, the problem asks whether the graph of a power grid can be observed by exhaustively applying two observation rules~\cite{brueni1993MinimalPP}:
First, every sensor observes its vertex and all neighbors.
Secondly, if a vertex is observed and has only one unobserved neighbor, that neighbor becomes observed, too.

\Pds is unfortunately NP-complete~\cite{brueni1993MinimalPP,haynes2002domination,kneis2006parameterized}, i.e., we cannot expect there to be an algorithm that performs reasonably on all inputs.
Moreover, the problem remains hard for a wide range of different graph classes~\cite{brueni1993MinimalPP,brueni2005ThePP,haynes2002domination,kneis2006parameterized,xu2006block,guo2008algorithms,liao2013circulararc}.
In terms of parameterized complexity, \Pds is known to be $W[2]$-hard~\cite{guo2008algorithms} when parameterized by solution size.

On the positive side, various approaches for solving \Pds have been proposed.
Theoretic results show that \Pds can be solved in linear time in graphs with fixed tree-width~\cite{kneis2006parameterized,guo2008algorithms}.
However, those algorithms have, to the best of our knowledge, never been implemented and are probably infeasible in practice due to their bad scaling with respect to the tree-width.
Several exponential-time algorithms have been presented~\cite{brueni1993MinimalPP,binkele2012exact} but those algorithms have not been implemented and evaluated either.

Practically feasible approaches using an MILP formulation have been proposed by Aazami~\cite{aazami2008DominationIG}.
This formulation was later improved upon by Brimkov, Mikesell, and Smith~\cite{brimkov2019connected} and most recently Jovanovic and Voss~\cite{jovanovic2020fixed}.
A different approach is to reduce \Pds to the hitting set problem~\cite{bozeman2018RestrictedPD,smith2020OptimalSP}.
This approach is based on the observation that one can determine so-called \emph{forts}, which are subsets of vertices that prevent propagation if none of them is selected.
A set of vertices is a valid solution for \Pds if and only if at least one vertex is selected for each fort, i.e., if it is a hitting set for the collection of all forts.
Graphs may contain an exponential number of forts, so this hitting set instance is not computed explicitly.
Instead, one can use the so-called \emph{implicit hitting set} approach, where one starts with a subset of all forts, computes a hitting set for this subset, and then validates whether this is already a solution for the \Pds instance.
If not, one obtains at least one new fort that can be added to the set of considered forts.
This is iterated until a solution is found.
This implicit hitting set approach has been used for other problems, e.g., for \textsc{MaxSAT}~\cite{saikko2016LMHSAS} and \textsc{TQBF}~\cite{janota2015SolvingQB}.
For \Pds, it has been introduced by Bozeman et al.~\cite{bozeman2018RestrictedPD}.
The strategy of finding forts has been later improved by Smith and Hicks~\cite{smith2020OptimalSP}, providing the current state-of-the-art for solving \Pds in practice.

Our contribution is threefold.
First, we study the parameterized complexity of \Pds parameterized by the solution size.
Though it is known to be $W[2]$-hard~\cite{guo2008algorithms}, it was unknown whether \Pds is also contained in $W[2]$.
We show that \Pds is $W[P]$-complete via a reduction from \WCS for circuits of arbitrary weft.
This completely determines its parameterized complexity and in particular shows that it is not in $W[2]$ unless $W[2] = W[P]$.
In our second contribution, we propose a set of reduction rules for pre-processing \Pds instances.
Our reduction rules aim to produce equivalent instances that are smaller and annotated with partial decisions, i.e., some vertices are marked as selected or as forbidden-to-select.
Though these annotations lead to a more general problem than the basic \Pds, we show that existing approaches for solving \Pds can be easily adapted to solve the annotated instances.
Moreover, we show that their performance greatly benefits from our reduction rules.
Finally, our third contribution is an improved heuristic for finding forts for the implicit hitting set formulation.
This improved heuristic together with our reduction rules beats the current state of the art solvers by more than one order of magnitude.
Moreover, our approach can solve previously unsolved instances of continental scale.

The remainder of this paper is organized as follows.
\Cref{sec:definitions} provides an overview of the basic concepts and notation used throughout this paper.
In \Cref{sec:wp complete}, we show that \Pds is $W[P]$-complete.
Our reduction rules and the heuristic for extending the hitting set instance are presented in \Cref{sec:solving}.
\Cref{sec:experiments} contains our experimental evaluation of the new method using a set of benchmark instances.

  \section{Preliminaries}
\label{sec:definitions}

\subparagraph{Graphs and Neighborhoods.}

Let $G = (V, E)$ be an undirected graph with vertices $V$ and edges $E$.
For $v \in V$, let $N(v) = \set{u \in V \mid uv \in E}$ be the \emph{open neighborhood} of $v$.
Similarly, $N[v] = N(V) \cup \set{v}$ is the \emph{closed neighborhood} of $v$.
Given a set $S \subseteq V$ we denote by $N(S)$ and $N[S]$ the union of all open and closed neighborhoods of the vertices in $S$.

\subparagraph{Power Dominating Set.}

For a given graph $G$, the problem \emph{\PDS (\Pds)} is to find a minimum vertex set $S \subseteq V$ of \emph{selected} vertices such that all vertices of the graph are observed.
We call such set a \emph{power dominating set}.
The size of a minimum power dominating set of a given graph $G$ is called the \emph{power dominating number} $\gamma_P(G)$.
Whether a vertex is \emph{observed} is determined by the following rules, which are applied iteratively.
We note that for the second rule, vertices can be marked as \emph{propagating}, i.e., the input of \Pds is not just a graph but a graph together with a set of propagating vertices.

\begin{description}
\item[Domination rule.]
  A vertex is observed if it is in the closed neighborhood of a selected vertex.
\item[Propagation rule.]
  Let $u \in V$ be a propagating vertex.
  If $u$ is observed and $v \in N(u)$ is the only neighbor of $u$ that is not yet observed, then $v$ becomes observed\footnote{
    The propagation rule is motivated by Kirchoff's law and Ohm's law in electric transmission networks.
    Propagating vertices are also called \emph{zero-injection vertices}.
    In electric networks, they refer to buses in substations without power injection, i.e. without attached loads or generators.
  }.
  If the propagation rule is applied to an observed vertex $u$, we say it \emph{propagates} its observation status.
\end{description}

The special case where we have no propagating vertices yields the well known \emph{\DS (\Ds)} problem.
Moreover, we refer to the special case where all vertices are propagating as \emph{\textsc{simple}-\Pds}.
In addition to the above \DS variants, we also consider the extension variant \Extension{\DS}.
For \extension{\Ds}, the input consists of the graph $G = (V, E)$, a set $X \subseteq V$ of \emph{pre-selected} and a set $Y$ of \emph{excluded} vertices; vertices in $V \setminus X \setminus Y$ are called \emph{undecided}.
\extension{\Ds} asks whether there is a solution $S \subset V$ such that $S$ includes all selected and excludes all excluded vertices, i.e., $X \subseteq S$ and $Y \cap S = \emptyset$.
The problems \extension{\Pds} and \extension{\textsc{simple}-\Ipds} are defined analogously.

\subparagraph{Hitting Set.}

Let $V$ be a set and let $\mathcal F \subseteq 2^V$ be a family of subsets.
A set $H \subseteq V$ is a \emph{hitting set} if it \emph{hits} every set $F \in \mathcal F$, i.e., $F \cap H \neq \emptyset$ for all $F \in \mathcal F$.
The problem \emph{\HS} is to find a hitting set of minimum size.
Note that the extension variant of \HS reduces to an instance of \HS itself, as one can simply remove excluded elements and remove the sets containing pre-selected elements.

\subparagraph{Parameterized Complexity}
\label{sec:param-compl}

We only give a very brief introduction; for more details, see one of the text books on parameterized complexity~\cite{downey2013fundamentals}.
For a parameterized problem, we are given a parameter $k$ in addition to the instance.
The running time is then not only analyzed in terms of the input size but also in terms of $k$.
We consider all problem variants introduced above with their canonical parameterization, i.e., parameterized by their solution size.
A parameterized problem is \emph{fixed parameter tractable (FPT)} if it can be solved in $f(k) \cdot n^{O(1)}$ where $f$ is a computable function.
\begin{description}
\item[{The $W$-Hierarchy.}]
To show that a problem is probably not in FPT, one can show hardness in terms of the $W$-hierarchy.
The \emph{$W$-hierarchy} consists of complexity classes $\mathrm{FPT} \subseteq \W{1} \subseteq \W{2} \subseteq \dots \subseteq \W{P}$, where each of the inclusion is assumed to be strict.
Many graph problems are known to be complete for $\W{1}$ and $\W{2}$, e.g., \textsc{Independent Set} and the above mentioned \DS are $\W{1}$- and $\W{2}$-complete, respectively.
We will see that the other variants of \Pds defined above are $\W{P}$-complete.

\item[{Proving $W[P]$-Hardness.}]
To show that a problem is $\W{P}$-hard, we need to reduce to it from another $\W{P}$-hard problem using a \emph{parameterized reduction}, i.e., a reduction that runs in FPT-time such that the change in the parameter is independent of the input size.
The $\W{P}$-hard problem we reduce from is \emph{\WMCS (\Wmcs)}, which is defined as follows.
A \emph{Boolean circuit} is a directed acyclic graph with a unique sink (the output node), where the sources are inputs and the inner nodes are logic gates (\AND, \OR, \NOT).
This defines a boolean function mapping the values on the input nodes to a Boolean output in the canonical way.
The problem \emph{\WCS (\Wcs)} with parameter $k$ asks whether there is a satisfying assignment that sets $k$ inputs to \TRUE and all other to \FALSE.
\emph{\WMCS (\Wmcs)} is the same problem with restriction that there are no \NOT-gates.
It is known that \Wmcs is $W[P]$-complete~\cite{downey2013fundamentals}.

\item[{Inclusion in $W[P]$.}]
Conversely, to show that a problem is in $\W{P}$, we use the following theorem.
\begin{theorem}[{\cite[Thm. 3.7]{cai1995structure}}, \cite{kneis2006parameterized}]
  \label{thm:wp contained ntm}
  Let $L$ be a parameterized problem in \NP.
  Then $L \in \W{P}$ if and only if there is a non-deterministic Turing machine deciding $L$ in at most $f(k) \log \abs{x}$ non-deterministic and $(\abs{x} + k)^{O(1)}$ deterministic steps, where $x$ is the input, $k$ the parameter, and $f$ a computable function.
\end{theorem}
\end{description}

  \section{Power Dominating Set is $W[P]$-Complete}
\label{sec:wp complete}

We prove $W[P]$-completeness via a chain of parameterized reductions from the \WMCS (\Wmcs) problem.
\Wmcs has a monotone Boolean circuit as input and asks whether it can be satisfied by setting at most $k$ inputs to \TRUE, where $k$ is the parameter.
We assume familiarity with the $W$-hierarchy and parameterized reductions; for a brief introduction, see \Cref{sec:param-compl}.
We start by introducing a variant of the \Pds problem that we use as an intermediate problem in our chain of reductions.

The input of the problem \emph{\IPDS (\Ipds)} is an instance of \Pds with the following additional information.
First, edges of the graph $G = (V, E)$ can be marked as \emph{booster edges}.
Secondly, we are given a set of \emph{implication arcs} $A \subseteq V \times V$.
We interpret $A$ as a set of directed edges on $V$ but perceive them as separate from the graph $G$, i.e., they do not affect the neighborhood.
In addition to the domination and propagation rule introduced in \Cref{sec:definitions}, we define the following to observation rules.

\begin{description}
\item[Booster rule.]
  Let $uv \in E$ be a booster edge.
  If $u$ is observed, then $v$ becomes observed and vice versa.
\item[Implication rule.]
  Let $(u, v) \in A$ be an implication arc and let $u$ be observed.
  Then $v$ also becomes observed.
\end{description}

We note that \Ipds is a generalization of \Pds in the sense that every \Pds instances is an instance of \Ipds with no booster edges and an empty set of implication arcs.
The extension variant \extension{\Ipds} is defined analogously to \extension{\Pds}.
Proving containment of \extension{\Ipds} in $W[P]$ is straight forward by giving an appropriate non-deterministic Turing machine.
We note that the analogous statement has been observed before for \Pds by Kneis et al.~\cite{kneis2006parameterized}.
Observe that this implies containment in $W[P]$ for all other problem variants we defined.

We first note that proving containment of \extension{\Ipds} in $W[P]$ is straight forward by giving an appropriate non-deterministic Turing machine.
The analogous statement for \Pds has been observed before by Kneis et al.~\cite{kneis2006parameterized}.

\begin{restatable}{lemma}{pdsInWP}
  \Extension{\IPDS} is in $\W{P}$.
\end{restatable}
\begin{proof}
  A nondeterministic Turing machine can guess a solution $S$ of size $k$ in $k \log \abs{V}$ nondeterministic steps.
  It can then check whether all vertices in $G$ are observed by $S$ in polynomial time.
  Thus, by \cref{thm:wp contained ntm}, \extension{\Ipds} is in $W[P]$.
\end{proof}
As all other variants of the power dominating set problem we consider are special cases of \extension{\Ipds}, this also proves containment in $W[P]$ for the other variants.

\subsection{Power Dominating Set to Simple Power Dominating Set}
\label{sec:pds to simple pds}

\begin{figure}
  \centering
\begin{tikzpicture}
  \def\extension#1{\textsc{#1-e}}
  \graph[grow right=4cm,branch down=1.5cm,simple] {
    wmcs/"{\Wmcs}" -> {%
      ipdse/"\extension{\Ipds}" -!- pdse/"\extension{\Pds}" -!- spdse/"\extension{\textsc{s-\Pds}}",
      ipds/\Ipds ->["{\cref{lem:ipds to pds}}"] pds/\Pds ->["{\cref{lem:pds to spds}}"] spds/\textsc{s-\Pds}
    };
    wmcs -!- ipds;
    wmcs ->["{\cref{lem:wmcs to ipds}}"] ipdse;
    ipdse ->["{\cref{lem:ipdse to ipds}}"'] ipds;
    {pds, spds} ->[densely dashed] {pdse, spdse};
  };
\end{tikzpicture}
  \caption{Reduction steps to show that \Pds and its variants are $W[P]$ hard. The solid arrows indicate our parameterized reductions described in this section. The hardness of the extension problems follows from the hardness of their basic problems, indicated by the dashed arrows.}
  \label{fig:hardness overview}
\end{figure}
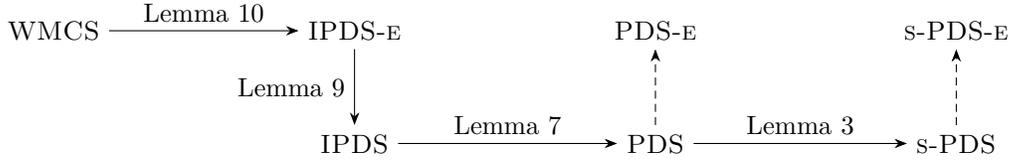

Our chain of reductions to prove $W[P]$-hardness is illustrated in \Cref{fig:hardness overview}.
We start with the reduction from \Pds to \textsc{Simple \Pds}, which is similar to the proof of $\W{2}$ hardness by \Pds~\cite{kneis2006parameterized,guo2008algorithms}.
The core idea is to simulate a non-propagating vertex with a propagating vertex with an additional leaf attached.

\begin{restatable}{lemma}{PdsToSpds}
  \label{lem:pds to spds}
  There is a parameterized reduction from \PDS to \textsc{Simple \PDS}.
\end{restatable}
\begin{proof}
  Our proof is similar to the proof of \W{2} hardness of \PDS~\cite{kneis2006parameterized,guo2008algorithms}.
  Let $G=(V, E)$ be a \Pds-instance with non-propagating vertices $Z \subseteq V$ (i.e., $V \setminus Z$ are propagating).
  We build a \textsc{simple-\Pds}-instance $G'=(V \cup V', E \cup E')$ (i.e., $V \cup V'$ are all propagating) from $G=(V, E)$ by attaching a new leaf to each non-propagating vertex $v \in Z$.
  On an intuitive level, attaching a leaf to a vertex has two effects.
  First, it is never optimal to choose a leaf to be part of the solution.
  Second, a vertex with an attached leaf can never propagate to any vertex except the leaf.
  Thus, attaching a leaf to every non-propagating vertex has the desired effect.
  Note that this is a parameterized reduction as the parameter is not changed.
  To make this argument more formal, we show that $G$ has a solution of size $k$ if and only if $G'$ has a solution of size $k$.

  Let $S$ be a power dominating set of $G$, i.e., applying the domination rule and the propagation rule (restricted to propagating vertices) exhaustively observes every vertex in $V$.
  Now interpret $S$ as a solution for $G'$.
  Note that the applying the domination rule has the same effect as before.
  Additionally, as the application of the propagation rule to $G$ was restricted to propagating vertices, whose neighborhood did not change in $G'$, it can be applied in the same way to $G$.
  Thus, all vertices in $V$ will be observed in $G'$, too.
  Applying additional propagation rules also observes the vertices in $V'$.

  Now let $S$ be a power dominating set of $G'$.
  Note that it is never optimal to choose one of the new degree-1 vertices as part of the solution, i.e., we can assume without loss of generality that $S \cap V' = \emptyset$.
  As $S$ is a power dominating set of $G'$, applying the domination rule and the propagation rule exhaustively observes all vertices.
  Interpreting $S$ as a solution for $G$, we can apply the domination rule as for $G'$.
  Additionally the propagation rule can be applied in the same way for propagating vertices, i.e., vertices not in $Z$.
  If, for $G'$, the propagation rule is applied to a vertex in $u \in Z$, then the newly observed vertex must be the leaf attached to $u$ in the construction of $G'$.
  Thus, all applications of the propagation rule in $G'$ also happen in $G$ except for those observing vertices in $V'$.
  Hence, $S$ is also a solution in $G$ with non-observed vertices $Z$.
\end{proof}

\subsection{Implicating Power Dominating Set to Power Dominating Set}
\label{sec:ipds to pds}

The reduction from \Ipds to \Pds, works in two steps.
First, we show that we can eliminate implicating arcs by replacing each of them with the small gadget show in \Cref{fig:implication gadget}.
Using another gadget (shown in \cref{fig:booster gadget}), we eliminate booster edges in a similar way, yielding the reduction.

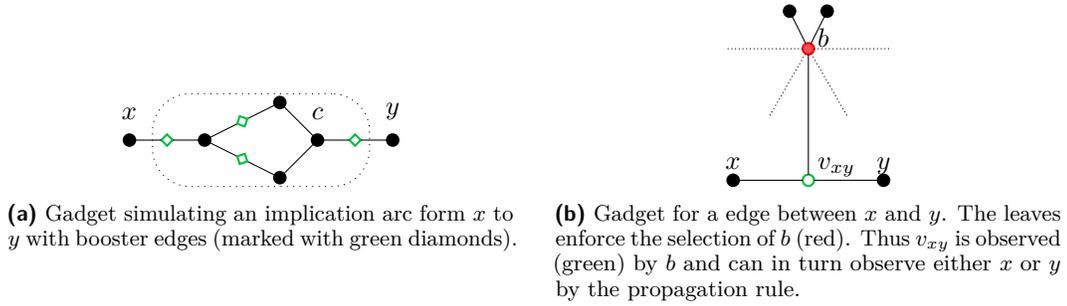
\begin{figure}
  \subcaptionbox{
    Gadget simulating an implication arc form $x$ to $y$ with booster edges (marked with green diamonds).
    \label{fig:implication gadget}
  }[.48\linewidth]{\tikzpicturename{gadget-implication-expanded}
\begin{tikzpicture}[pds instance]
	\graph[empty nodes,grow right=0.5] {[nodes={node,propagating}]
		x["$x$"] -- b1[booster]
		-- c1 -!- {[nodes={y=0.5}] b11[coordinate] -!- v11, b12[coordinate] -!- v12}
		-- c3["$c$"] -- b2[booster] -- y["$y$"]
	};
  \draw (c1) -- node[sloped,booster] {} (v11);
  \draw (c1) -- node[sloped,booster] {} (v12);
	\node[fit={(b1)(b11)(b12)(b2)},draw=black,dotted,rounded corners=5mm] {};
\end{tikzpicture}}
  \hfill
  \subcaptionbox{
    Gadget for a edge between $x$ and $y$.
    The leaves enforce the selection of $b$ (red).
    Thus $v_{xy}$ is observed (green) by $b$ and can in turn observe either $x$ or $y$ by the propagation rule.
    \label{fig:booster gadget}
  }[.48\linewidth]{\tikzpicturename{booster-expanded}
\begin{tikzpicture}[pds instance,every label/.style={anchor=base}]
	\graph[empty nodes] {
		b["$b$" anchor=west,node,solution,propagating,shift={(1, 0.75)}];
		{[nodes={node,propagating,baseline}]a["$x$"] -- v["$v_{xy}$" anchor=base west,observed,solution=no,propagating] -- c["$y$"]};
		v -- b;
    b -- {subgraph I_n[n=2,branch left=5mm,name=P,phase=60,radius=7mm,nodes={node,propagating,shift={([shift={(2.5mm,5mm)}]b.center)}}]};
    b --[densely dotted] {subgraph I_n[n=4,circular placement,nodes={draw=none,shift={(b.center)}},phase=180,radius=12mm]};
	};
\end{tikzpicture}}
	\caption{Gadgets for implication arcs and booster edges.}
\end{figure}

\begin{lemma}
  \label{lem:implication gadget}
  Every instance of \IPDS can be reduced to an equivalent instance with no implication arcs without changing the parameter.
\end{lemma}
\begin{proof}
  Let $G$ be an instance of \Ipds with implication arcs $A$ and let $a = (x, y) \in A$ be an implication arc.
  We show that replacing $a$ with the gadget $I_a$ in \cref{fig:implication gadget} yields an equivalent instance with one less implication arc.
  Applying this replacement to all implication arcs yields the claim.

  To explain this on an intuitive level, consider the implication gadget $I_{a}$ and first assume that $x$ is observed.
  It is not hard to verify that, by applying the booster and the propagation rules, all vertices in $I_a$ and $y$ are observed.
  Conversely, if only $y$ is observed, then $c$ cannot propagate due to its two unobserved neighbors.
  Thus, $I_a$ mimics the behavior of an implication arc.

  To formally prove the claim, let $G'$ be the instance obtained by replacing one implication arc $a$ with the implication gadget $I_a$.
  We show $G$ has a solution of size $k$ if and only if $G'$ has a solution of size $k$.
  For the first direction, let $S$ be a solution of $G$.
  Consider a corresponding sequence $r_1, \dots, r_\ell$ of observation rules.
  Let $V_i$ be the set of observed vertices after the application of the first $i$ rules $r_1, \dots, r_i$.
  Observe that the vertices in $S$ are only observed after the domination rule is applied and thus we define $V_0 = \emptyset$.
  Now interpret $S$ as a candidate solution for $G'$.
  Based on $r_1, \dots, r_\ell$, we give a sequence of observation rules on $G'$ that fully observes the graph.
  Our argument is of inductive nature, i.e., for fixed $i$, we assume that we have a rule sequence for $G'$ that observes at least the vertices in $V_{i - 1}$ and we show that it can be extended to observe the vertices in $V_i$ using a set of rules based on $r_i$.
  In the following we discriminate between the possible types of $r_i$.
  If $r_i$ is the domination or the booster rule, it can be applied as in $G$.
  If $r_i$ is the propagation rule and the propagating vertex is not $x$ or $y$, it can be applied as in $G$.
  Otherwise, note that the only new neighbors to $x$ and $y$ are connected via booster edges.
  Thus, if $x$ or $y$ are observed, we can apply the booster rule to observe their new neighbors.
  Afterwards, the propagation rule can be applied as usual, as the relevant vertex has again only one unobserved neighbor.
  If $r_i$ is the implication rule applied to an implication arc other than $a = (x, y)$, we can apply it as before.
  Otherwise, the above observation shows that all vertices in the implication gadget $I_a$ and $y$ will be observed, which observes at least all vertices in $V_i$.
  It thus follows that after step $\ell$ all vertices in $V_\ell = V$ will be observed in $G'$.
  Additionally the vertices inside the gadget $I_a$ can also be observed by applying observation rules knowing that $x \in V_\ell$ is observed.

  Conversely, assume that $S$ is a solution for $G'$.
  Without loss of generality, we assume that $S$ does not contain any vertices from $I_a$ as we can otherwise choose $x$ instead and apply booster and propagation rules to observe the whole gadget.
  Additionally let $r_1, \dots, r_\ell$ be a sequence of rule applications that observe the whole graph $G'$ and let $V_i$ again be the set of vertices observed after applying $r_1, \dots, r_i$ where $V_0 = \emptyset$.
  Now consider the lowest index $\iota$ such that $x \in V_\iota$.
  Then without loss of generality, assume that the rules following $r_\iota$ are booster and propagation rules that observe the whole gadget $I_a$ and $y$ before any other rules are applied.
  Now consider $S$ to be a solution candidate in $G$.
  We can apply the sequence of rules $r_1, \dots, r_\ell$ as in $G'$ except for the rules observing $I_a$ following $r_\iota$.
  However, this subsequence of rules can be replaced with a single application of the implication rule on $(x, y)$, which concludes the proof.
\end{proof}

The gadget for simulating booster edges shown in \cref{fig:booster gadget} requires adding a globally unique selected vertex $b$ to which all gadgets are connected.
We enforce that $b$ is selected by adding two leaves.
\begin{lemma}
  \label{lem:booster vertex}
  If an \Ipds-instance contains a vertex $b$ with two or more leaves that don't have any booster edges or implication arcs, there is a minimum power dominating set containing $b$.
\end{lemma}
\begin{proof}
  We show that at least one of $b$ or its leaves, $u$ and $v$  must be selected.
  First, assume that none of these vertices is selected.
  If $b$ is not propagating, there is no observation rule by which $u$ and $v$ can be observed.
  Otherwise the only candidate rule is the propagation rule which in turn cannot be applied to $b$ before at least one of the leaves is observed.
  Hence, every power dominating set must contain at least one of $b, u$ or $v$.

  In this case we can always select $b$ and observe the leaves with the domination rule.
\end{proof}

\begin{lemma}
  \label{lem:booster gadget}
  Every \Ipds-instance with booster edges can be reduced to an equivalent instance without booster edges.
\end{lemma}
\begin{proof}
  Let $G = (V, E)$ be an \Ipds-instance with booster edges and let $xy \in E$ be a booster edge.
  Our construction requires a known selected vertex.
  To ensure such a vertex exists, we insert a new vertex $b$ and attach two leaves to $b$.
  By \cref{lem:booster gadget} we can assume without loss of generality that $b$ is selected.

  We show that subdividing $xy$ by adding a third vertex $v_{xy}$ between $x$ and $y$ and connecting $v_{xy}$ to $b$ yields an equivalent instance $G'$ with one fewer booster edge.
  Replacing every booster edge in this way yields an instance without booster edges.
  This gadget construction is also shown in \cref{fig:booster full}.
  Note that we only need to insert $b$ once when replacing the booster edges iteratively.

  Intuitively, the gadget introduces a vertex $v_{xy}$ between $x$ and $y$ that is always observed.
  Then, if either $x$ or $y$ become observed at some point, $v_{xy}$ has only one unobserved neighbor left which becomes observed by applying the propagation rule.

  To formally prove the claim, we show that $G$ has a power dominating set of size $k$ if and only if $G'$ has a power dominating set of size $k + 1$.
  For the first direction let $S$ be a minimum power dominating set of $G$ and let $r_1,\dots,r_\ell$ be a sequence of observation rules the application of which observes all vertices.
  We show that $S \cup \set{b}$ is a power dominating set of $G'$.
  Based on $r_1,\dots,r_\ell$ we give a sequence of rule applications that observes all vertices in $G'$.
  We use an inductive argument, i.e. for a given fixed $i$ we assume that we have a sequence of observation rules $r'_1,\dots,r'_j$ for $G'$ that observes at least all vertices observed by $r_1,\dots,r_{i-1}$ in $G$.
  Based on $r_i$, we extend that sequence with observation rules $r'_{j+1},\dots,r'_{j'}$ to observe all missing vertices that become observed by applying $r_i$.
  The first rule $r'_1$ we apply in $G'$ is the domination rule on $b$.
  In each step we handle rules that are affected by the introduction of the gadget.

  The booster gadget does not affect the implication rule which can thus be applied as before.
  The other rules require adaptation if they are applied to $x$ or $y$.

  If $r_i$ if the domination rule applied to $x$, we use the domination rule in $G'$ and add an application of the propagation rule on $v$ to ensure $y$ becomes observed.
  The same applies to the domination rule on $y$.
  All other edges remain the same in the construction and thus the domination rule can be applied as before.

  If $r_i$ is the propagation rule and $x$ is the propagating vertex, propagating to $y$, we could instead use the booster rule, so we proceed like described there.
  The same applies if $y$ is the propagating vertex, observing $x$.
  All other application of the propagation rule can be applied unchanged in $G'$ due to our induction hypothesis.

  If $r_i$ is booster rule applied to $xy$ we replace it by an application of the propagation rule applied to $v_{xy}$.
  The inserted vertex $v_{xy}$ is observed in the application of the domination rule in $r'_1$.
  Applying the booster rule requires an observed endpoint, let this be $x$.
  By assumption, $x$ is also observed in $G'$ where $v$ has thus at most one unobserved neighbor, $y$.
  The propagation rule can thus be applied in $G'$ and $y$ becomes observed.
  Applications of the booster rule to other edges are applied in $G'$ in the same way as in $G$.

  The argument for the other direction follows the same structure with the roles of $G$ and $G'$ swapped.
  Let $S$ be a minimum power dominating set of $G'$ and let $r'_1,\dots,r'_\ell$ be a corresponding sequence of observation rules observing all vertices.
  Given $r'_1,\dots,r'_\ell$ we construct a sequence of rule applications that observes all vertices in $G$.

  We ignore all rules applied to $b$ and its attached leaves.
  With the exception of the propagation rule, all observation rules can be applied in $G$ in the same way as in $G'$.
  If $r'_i$ is the propagation rule we distinguish three cases by the propagating vertex.
  \begin{enumerate}
    \item The propagating vertex is $v_{xy}$.
      We know that $x$ or $y$ is observed; without loss of generality assume this is $x$.
      We apply the booster rule on $xy$ and $y$ becomes observed in $G$.
    \item The propagating vertex is $x$.
      We first apply the booster rule on $xy$ to make sure $y$ is observed in $G$.
      Now $x$ has the same unobserved neighbor in $G$ and $G'$ and thus the rule can be applied like before.
      The same applies if the propagating vertex is $y$.
    \item Otherwise, the propagation rule can be applied as before.
  \end{enumerate}
\end{proof}

\begin{restatable}{lemma}{IpdsToPds}
  \label{lem:ipds to pds}
  There is a parameterized reduction from \IPDS to \PDS.
\end{restatable}
\begin{proof}
  An \Ipds-instance $G$ may have implication arcs and booster edges, which are not allowed in \Pds instances.
  As shown in \cref{lem:implication gadget,lem:booster gadget}, we can construct an equivalent instance instance $G'$ which does not contain implication arcs or booster edges.
  Lacking those special  edges, $G'$ is also a \Pds instance.
  Thus, this construction is a parameterized reduction from \Ipds to \Pds, increasing the parameter, i.e. the solution size, by one in the process.
\end{proof}

\subsection{Extension to Non-Extension (for IPDS)}
\label{sec:ipdse to ipds}

A reduction from \extension{\Ipds} to \Ipds requires a mechanism to enforce the selection of certain vertices and to make sure other vertices are not selected in a minimum power dominating set.
We already saw a way to ensure that a vertex gets selected in \cref{lem:booster vertex}.
The core difficulty thus comes from enforcing the excluded vertices to not be selected.
The basic idea how we achieve this is by using many copies of the graph and an additional clique of non-propagating vertices.
The vertices in the clique represent the vertices that are allowed to be selected and provide the only connection between the copies.
With this construction, selecting vertices outside the clique is never optimal.

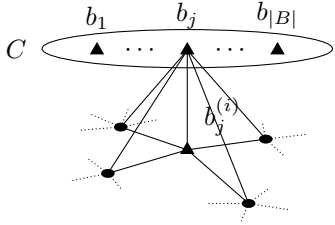
\begin{figure}
  \tikzpicturename{rpds}
\begin{tikzpicture}[pds instance,x=8mm,y=8mm]
	\path (0,0) node[input,label={$b_1$}] (x1) {} -- node[pos=0.25] {$\cdots$} node[input,pos=0.5,label={$b_j$}] (xi) {} node[pos=0.75] {$\cdots$} ++(3,0) node[input,label={$b_{\abs{B}}$}] (xp) {};
	\node[draw=black,shape=ellipse,fit={(x1)(xi)(xp)},"$C$" left] (c) {};
	\begin{scope}[node/.append style={shape=ellipse,minimum height=0.6\pgfplotmarksize}]
		\node[input,below=1.5 of xi,"30:$b_{j}^{(i)}$"] (y1) {};
		\node[node,above left=0.3 and 1.0 of y1] (y2) {};
		\node[node,above right=0.1 and 1.2 of y1] (y3) {};
		\node[node,below left=0.3 and 1.2 of y1] (y4) {};
		\node[node,below right=0.8 and 0.9 of y1] (y5) {};
	\end{scope}

	\foreach \v/\n/\exp in {
		y2/3/{55*\noexpand\i+25},
		y3/2/{80*\noexpand\i-65},
		y4/2/{85*\noexpand\i+120},
		y5/3/{80*\noexpand\i+180}
	}{
		\foreach \i in {0,...,\n}{
			\edef\exp{\exp}
			\pgfmathsetmacro\angle{\exp}
			\pgfmathsetmacro\x{cos(\angle) * 0.7}
			\pgfmathsetmacro\y{sin(\angle) * 0.3}
			\draw[densely dotted] (\v) -- ++(\x,\y);
		}
	}
	\graph[use existing nodes] {
		y1 -- {y2, y3, y4, y5};
		xi -- {y1, y2, y3, y4, y5};
	};
\end{tikzpicture}
  \caption{The vertices $b_j$ in the top-level clique $C$ connect to the neighborhood of their corresponding vertices $b^{(i)_j}$ in each copy $G^{(i)}$ of $G$.}
  \label{fig:restricted pds construction}
\end{figure}

\begin{lemma}
  \label{lem:vertex exclusion}
  Every instance of \extension{\Ipds} can be reduced to an equivalent instance without excluded vertices without changing the parameter.
\end{lemma}
\begin{proof}
  Let $G = (V, E)$ be an \extension{\Ipds}-instance and let $Y$ be a set of vertices excluded from a solution to $G$.
  We construct an \extension{\Ipds} instance $G'$ where no vertices are explicitly excluded.

  Let $B = \set{b_1,\dots,b_{\abs{B}}} = V \setminus Y$ be the set of vertices allowed in a solution.
  We initialize $G'$ with $\abs{V} + 1$ copies of $G$ which we denote by $G^{(i)}$ and a clique $C$ consisting of the vertices in $B$.
  Each copy $G^{(i)}$ consists of a copy of the vertices $V^{(i)} = \set{v^{(i)} \mid v \in V}$.
  These copied vertices are connected by edges if their original counterparts are connected, i.e. $E^{(i)} = \set{v^{(i)}w^{(i)} | vw \in E}$.
  Each vertex $b_j$ in the clique has edges to all vertices in the neighborhood of its counterparts in the copies, i.e. there is an edge from $b_j$ to every vertex in $N[x^{(i)}]$ in every copy $i$.
  Each vertex $v^{(i)}$ in the copy $G^{(i)}$ is propagating if and only if $v$ is propagating in $G$.
  All vertices in $C$ are non-propagating.
  See \cref{fig:restricted pds construction} for an example.

  This construction ensures that vertices selected in one $G^{(i)}$ never observe vertices in other copies or cause them to become observed by the propagation rule.
  If a minimum power dominating set contains a vertex in one $G^{(i)}$, it must therefore also contain vertices from each other $G^{(i')}$.
  Hence, a minimum power dominating set of size less than or equal to $\abs{V}$ cannot contain any vertices outside $C$.

  It remains to show that $G$ has an implicating power dominating set of size $k$ if and only if $G'$ has a power dominating set of size $k$.
  Let $S \subseteq V$ be a minimum power dominating set of $G$ and let $r_1,\dots,r_\ell$ be a sequence of observation rules that observes all vertices in $G$.
  It is easy to verify that $S$ is also a power dominating set of $G'$ by modifying the sequence $r_j$.
  If $r_j$ is an application of the domination rule, we apply it to the same vertex in $G'$ where it is part of the clique.
  By definition, $S$ contains only vertices in $B$.
  Otherwise, $r_i$ is some other observation rule and we apply it in every copy $G^{(i)}$ in the same way as in $G$.
  The first application of a domination rule observes all vertices in $C$.
  Thus, after the transformed sequence of observation rules is applied, all vertices in $G'$ are observed.

  Now let $S$ be a minimum power dominating set of $G'$.
  It follows from the previous direction that, if a minimum power dominating set of $G'$ has more than $\abs{V}$ vertices, $G$ does not have a power dominating set.
  We show by contradiction that $S$ can only contain vertices in the top-level clique $C$.
  Assume $S \setminus B \neq \emptyset$.
  Because $\abs{S} \leq \abs{V(G)}$, there is at least one $G^{(\hat{\iota})}$ without any vertices in $S$.
  By out assumption that $S$ is a power dominating set, all vertices in $G^{(\hat{\iota})}$ are observed by $S$.
  All vertices in $C$ are non-propagating and have no booster edges or implication arcs.
  This means that $S \cap B$ is a power dominating set of the induced subgraph in $G'$ by $B$ and $V(G^{(\hat{\iota})})$.
  But then, $S \cap B$ is a power dominating set for $G'$ as all $G^{(i)}$ are identical.
  This contradicts our assumption that $S$ is a minimum power dominating set and hence no minimum power dominating set can have vertices not in $B$.
\end{proof}

\begin{restatable}{lemma}{IpdseToIpds}
  \label{lem:ipdse to ipds}
  There is a parameterized reduction from \Extension{\IPDS} to \IPDS.
\end{restatable}
\begin{proof}
  \extension{\Ipds} differs from \Ipds by requiring the solution to contain certain vertices while excluding others.
  We can use \cref{lem:booster vertex} to force the inclusion of vertices in the solution by adding two leaves.
  \Cref{lem:vertex exclusion} provides a method to enforce the exclusion of vertices.
  Together the yield a parameterized reduction from \extension{\Ipds} to \Ipds.
\end{proof}

\subsection{WMCS to IPDS-Extension}
\label{sec:wmcs to ipds}

So far we only considered different extensions of the \PDS problem.
We will now see that those extensions allow a straightforward reduction from \WMCS, a $W[P]$ complete problem.
The core idea is to replace the arcs in the directed acyclic graph describing the circuit with implication arcs and to model \AND-gates using a gadget as show in \Cref{fig:and gadget}.
This last step in our reduction chain finally yields the $W[P]$-hardness of \Pds.

\begin{figure}
  \subcaptionbox{
    Input monotone circuit consisting of \texttt{and}-gates and \texttt{or}-gates\label{fig:monotone circuit}
  }[.27\linewidth]
  {\tikzpicturename{reduction-circuit}
\begin{tikzpicture}[pds instance,y=0.923cm]
\graph[grow down,branch right,nodes={node},empty nodes]{
	subgraph I_n [n=4,name=In,nodes={input}] -!-
	subgraph I_n [n=3,name=Or1,nodes={or,x=0.5}] -!-
	subgraph I_n [n=2,name=And1,nodes={and,x=1}] -!-
	subgraph I_n [n=3,name=Or2,nodes={or,x=0.5}] ->
	subgraph I_n [n=1,name=And2,nodes={and,x=1.5}] ->
	out[output,x=1.5];
	{In 1, In 2} -> Or1 1;
	{In 2, In 3} -> Or1 2;
	{In 3, In 4} -> Or1 3;
	{Or1 1, Or1 2} -> And1 1;
	{Or1 2, Or1 3} -> And1 2;
	And1 1 -> Or2 1;
	{And1 1, And1 2} -> Or2 2;
	{And1 1, And1 2} -> Or2 3;
};
\end{tikzpicture}}
  \hfill
  \subcaptionbox{
    Corresponding \Ipds instance.
    Only the triangular vertices in the top clique are selectable.\label{fig:ipds circuit}
  }[.48\linewidth]{\tikzpicturename{reduction-pds}
\begin{tikzpicture}[pds instance,l/.style={bend left=#1},r/.style={bend right=#1},
	default to/.style={to path={-- (\tikztotarget) \tikztonodes}},
	->o-/.style={to path={-- node[node,propagating,pos=0.66] (tmp) {} (\tikztotarget)\pgfextra{\draw[->] (\tikztostart) -- (tmp);}}},
	x=0.5cm,y=0.33cm,
]
	\graph[grow down=2,branch right=2,color class=ins1,color class=ins2]{
		{[nodes={node},empty nodes]
			{
				subgraph I_n [n=4,name=In,nodes={input},ins2]
				-!- subgraph I_n [n=3,name=Or1,nodes={or,x=1}]
				-!- subgraph I_n [n=2,name=And1in,nodes={and,x=2,y=-0.8}]
				-- subgraph I_n [n=2,name=And1out,nodes={and,x=2}]
				-- subgraph I_n [n=3,name=Or2,nodes={or,x=1}]
				--[->o-] subgraph I_n [n=1,name=And2,nodes={and,x=3,y=-0.8}]
				-- and[and,x=3]
				-- out[output,x=3];
				{In 1, In 2} -> Or1 1;
				{In 2, In 3} -> Or1 2;
				{In 3, In 4} -> Or1 3;
				{Or1 1, Or1 2} --[->o-] And1in 1;
				{Or1 2, Or1 3} --[->o-] And1in 2;
				And1out 1 -> Or2 1;
				{And1out 1, And1out 2} -> Or2 2;
				{And1out 1, And1out 2} -> Or2 3;
				out ->[l=2.1cm] In 1; out ->[l=2.7cm] In 2;
				out ->[r=2.1cm] In 4; out ->[r=2.7cm] In 3;
			}
		};
	};
	\begin{scope}[on background layer]
		\node[fit={(In 1) (In 4)},draw=black,rounded corners,inner sep=2mm,thick,densely dotted] {};
	\end{scope}
	\end{tikzpicture}}
  \hfill
  \subcaptionbox{
    \AND-gadget\label{fig:and gadget}
  }[.19\linewidth]{\tikzpicturename{reduction-and}
\begin{tikzpicture}[pds instance,y=0.75cm,->o-/.style={to path={-- node[node,propagating,pos=0.66] (tmp) {} (\tikztotarget)\pgfextra{\draw[->] (\tikztostart) -- (tmp);}}}]
	\path (0,-4) graph [empty nodes,grow down=1,branch right=1] {
		subgraph I_n[n=3]
		-> and out [and,x=1]
		-> subgraph I_n[n=3,name=below]
	};
	\path[shift={(0,-3.7)}] graph [empty nodes,grow down=1,branch right=1] {
		subgraph I_n[n=3,name=above]
		--[->o-] and in [and,x=1,y=-0.2,"$x$" right,anchor=base]
		-- and out [and,x=1,y=0.2,"$y$" right,anchor=base]
		-> subgraph I_n[n=3]
	};
	\draw[->,decoration={snake,amplitude=.4mm,segment length=1.5mm,post length=2mm},decorate] (below 2) -> (above 2);
\end{tikzpicture}}
  \caption{Example of a transformation from a circuit to an \Ipds instance. We transform the circuit in \subref{fig:monotone circuit} to the semi-directed graph in \subref{fig:ipds circuit} by replacing the \AND-gates with the gadget in \subref{fig:and gadget} and inserting implication edges from the output to each input.}
  \label{fig:wmcs to ipds}
\end{figure}
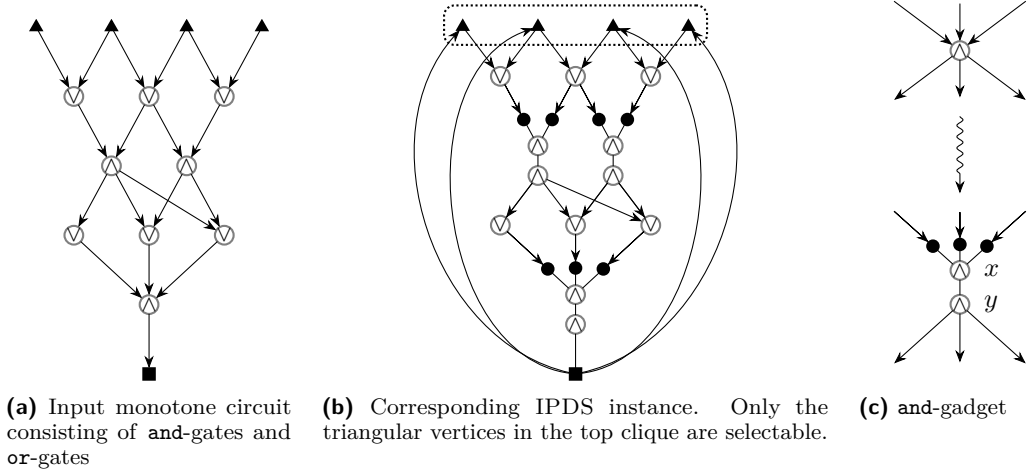
\begin{restatable}{lemma}{WmcsToIpds}
  \label{lem:wmcs to ipds}
  There is a parameterized reduction from \WMCS to \Extension{\IPDS}.
\end{restatable}
\begin{proof}
  Let $C = (V, E)$ be a monotone circuit interpreted as acyclic graph with input nodes $\Vin$, \AND-gates $\Vand$, \OR-gates $\Vor$ and the output node $\texttt{out}$.
  \Cref{fig:monotone circuit} show an example of such a circuit.
  Its corresponding \extension{\Ipds}-instance is depicted in \cref{fig:ipds circuit}.
  To construct such an instance from $C$ we proceed as follows.
  We interpret all directed edges $E$ in $C$ as implication arcs $A$ and add further implication arcs from the output node \texttt{out} to every input.
  Next we replace every \AND-gate $v$ by two new connected vertices $x$ and $y$ where all outgoing edges of $v$ are instead outgoing implication arcs of $y$; see \cref{fig:and gadget} for an example.
  For every incoming edge of $v$ from a vertex $u$, we place a fresh vertex $x_u$ and add an edge $x_u x$ and an implication arc $(u, x_u)$.
  We refer to $x$ as the \emph{input} of the gate and to $y$ as the \emph{output} of the gate.
  We call the fresh nodes $x_u$ \emph{proxy inputs}.
  They are marked in gray in \cref{fig:and gadget}.
  The other gates and vertices are their own input and output.
  We exclude all vertices except the inputs from a solution to the \extension{\Ipds} instance.

  The basic idea of the construction is that we interpret true values in the circuit as a node being observed in $G$.
  We can then simulate \AND-gates by exploiting the propagation rule and \OR-gates using the implication rule.
  Whenever a single input node of an \OR-gate becomes observed, the implication rule ensures that first the gate and then all its children become observed, too.
  Each proxy input propagates its observation to the gate input $x$, which in turn can only propagate its observation to the gate output $y$ when all inputs are observed.

  It remains to show that $C$ has a satisfying assignment with at most $k$ true variable if and only if $G$ has a power dominating set of size $k$.
  \begin{claim*}[A satisfying assignment of $C$ implies a power dominating set in $G$]

    Let $\mathcal{X}$ be a satisfying assignment with true variables $v_1,\dots,v_k$ and let $x_1,\dots,x_\ell$ be a topological ordering of the nodes and gates with \TRUE output.
    We note that the subgraph induced in $C$ by $\set{x_1,\dots,x_\ell}$ must be connected and contains a path from the input nodes to the output.
    A candidate solution for the \extension{\Ipds}-instance $G$ is $S = \set{x_1,\dots,x_\ell}$.
    We use a two step argument two show that $S$ does indeed observe all vertices in $G$.
    First, we show inductively, that if the vertices in $S$ are observed, the output node becomes observed.
    When the output node is observed, all inputs become observed by the implication arcs from the output node.
    The second step is to show that this suffices to observe all remaining vertices.

    To show that $S$ does observe the output node, we use an inductive proof.
    It is based on the assumption that for any fixed $i$, we can find a sequence of observation rules that observes at least the vertices in $G$ corresponding to the \TRUE gates $x_1,\dots,x_{i-1}$.
    To extend the sequence to $x_i$, we differentiate by the type of $x_i$:
    \begin{enumerate}
      \item If $x_i$ is an input, it must be in $\mathcal{X}$ and thus is selected and thus observed by the domination rule.
      \item If $x_i$ is an \OR-gate, there is some \TRUE gate $x_j$ with $j < i$.
      By our induction hypothesis, we know that the corresponding vertex in $G$ is observed and we can use the implication rule on the implication arc to $x_i$.
      \item If $x_i$ is an \AND-gate, all gates with edges into $x_j$ are \TRUE and their corresponding vertices are thus observed.
      We first apply the implication rule on the implication arcs to the proxy inputs, then use the propagation rule to observe the gate input.
      The gate input thus has only one unobserved neighbor left, the gate output, which becomes observed by the propagation rule.
      \item If, finally, $x_i$ is the output node, we apply the implication rule on its parent node, observing $x_i$.
    \end{enumerate}

    After the output node is observed, we can use the implication rule on the implication arcs to the inputs, observing all of them.
    Note that in a monotone circuit with all inputs set to \TRUE, all gates and the output are \TRUE, too.
    Then, by the same inductive argument as above, all vertices in $G$ become observed.
  \end{claim*}

  \begin{claim*}[A power dominating set of $G$ implies a satisfying assignment of $C$]
    Let $S$ be a minimum power dominating set of $G$.
    By construction, all vertices in $S$ are input nodes and thus we obtain a candidate solution $\mathcal{X}$ where all inputs in $S$ are set to \TRUE and all others to \FALSE.
    It remains to show that $\mathcal{X}$ is indeed a solution of $C$.
    Let $r_1,\dots,r_{\hat{\ell}}$ be a sequence of observation rules for $S$ that observes all vertices in $G$ and let $V_i$ be the set of vertices observed after the application of rule $i$.
    We only consider the subsequence $r_1,\dots,r_\ell$ where $r_{\ell+1}$ is the first application of the implication rule on an implication arc from \texttt{out} to an input vertex.

    We claim that every gate in $C$ with an observed output node in $V_\ell$ outputs \TRUE when evaluating the circuit.
    Assume the claim is false and let $\hat{\iota} \leq \ell$ be the smallest index such that $V_{\hat{\iota}}$ contains a gate output vertex $v$ whose gate $x_v$ in $C$ outputs \FALSE.
    We consider the different gate types of $x_v$ and show that this leads to a contradiction.
    \begin{enumerate}
      \item If $x_v$ is an input, we know that $x_v \in \mathcal{X}$.
      Otherwise, the only way to observe $x_v$ is by an implication arc from \texttt{out}, which we explicitly excluded from $r_1,\dots,r_\ell$.
      \item If $x_v$ is an \OR-gate or \texttt{out}, $v$ has at least one observed predecessor and became observed by the implication rule.
      Because $\hat{\iota}$ is minimal, $x_v$ also has a \TRUE input and therefore outputs \TRUE.
      \item If $x_v$ is an \AND-gate, the gate input and the proxy inputs must all be observed.
      Each proxy input $v_u$ has an incoming implication arc from the output vertex $u$ of some other logic gate.
      By definition of $V_{\hat{\iota}}$, we know that $u$ is observed and thus $x_u$ outputs \TRUE.
      So all inputs of $x_v$ are true and it outputs \TRUE.
    \end{enumerate}
  \end{claim*}
  The output \texttt{out} is thus \TRUE, concluding the proof.
\end{proof}

This proof concludes the chain of reductions as presented in \cref{fig:hardness overview}.
We saw that we can reduce from the $W[P]$-complete \WMCS to \Extension{\IPDS} in \cref{lem:wmcs to ipds}.
\extension{\Ipds}, in turn, can be represented in terms of \IPDS, as seen in \cref{lem:ipdse to ipds}.
\Cref{lem:ipds to pds} showed that booster edges and implication arcs can be replaced by gadgets with only normal vertices.
Propagating vertices can also be eliminated as we showed in \cref{lem:pds to spds}.
All these transformations are parametric reductions and thus \textsc{Simple-\Pds} is at least as hard as \Wmcs, i.e. both are $W[P]$ hard.

\begin{restatable}{corollary}{pdsIsWPComplete}
  \label{lem:pds is wp complete}
	\PDS is \W{P} complete.
\end{restatable}
\begin{proof}
  By \cref{thm:wp contained ntm}, \extension{\Ipds} is in $W[P]$.
  Being a special case of \extension{\Ipds}, \Pds it is also contained in $W[P]$.
  The $W[P]$-hardness follows from the reduction chain in \cref{lem:wmcs to ipds,lem:ipdse to ipds,lem:ipds to pds,lem:pds to spds}.
\end{proof}


\section{Solving \PDS}
\label{sec:solving}

In this section, we give an algorithm for solving \extension{\Pds}.
Our algorithm consists of different phases.
In the first phase, we apply the reduction rules described in \cref{sec:reduction rules}.
Each rule either shrinks the graph or decides for a vertex that it should be pre-selected or excluded.
We prove that the rules are safe, i.e., they yield equivalent instances.
Afterwards, in \cref{sec:sub-problems} we split the instance into several components that can be solved independently.
Finally, each of these subinstances is solved exactly using the implicit hitting set approach~\cite{bozeman2018RestrictedPD} with our improved strategy for finding new sets that need to be hit; see Section~\ref{sec:kernel solver}.

We note that these phases are somewhat modular in the sense that one could easily add further reduction rules or that one can replace the algorithm for solving the kernel in the final step.
In our experiments in Section~\ref{sec:experiments}, we also use an MILP for this step.
This MILP formulation is based on the formulation for \Pds by Jovanovic and Voss~\cite{jovanovic2020fixed} and we discuss our adjustments in \cref{sec:milp formulation}.
Moreover, instead of solving the subinstances optimally, one can instead use a heuristic solver.
In the latter case, the preceding application of our reduction rules and splitting into subinstances then helps to find better solutions rather than improving the running time of exact solvers.
This is used in our experiments to find upper bounds on the power domination number before we have an exact solution.

\subsection{Reduction Rules}
\label{sec:reduction rules}

\makeatletter
\newcommand{\reductionname}[1]{
  \let\oldthereduction\thereduction
  \renewcommand{\thereduction}{\texttt{#1}}
  \g@addto@macro\endreduction{
  \global\let\thereduction\oldthereduction}
}
\makeatother

Many of our reduction rules are local in the sense that they transform one substructure into a different substructure.
Most of our local reduction rules are illustrated in \Cref{fig:local-red-rules}.
For proofs of safeness, see \Cref{sec:reduct-rules-appendix}.

\begin{figure}
  \begin{subcaptionblock}{\linewidth}
    \centering
    \tikzpicturename{rules-leaf1}
\begin{tikzpicture}[pds instance,every node/.style={node}]
  \begin{scope}[xshift=2cm]
    \graph[grow right,branch down,nodes={every node},empty nodes] {
      a[propagating=dontcare,solution=undecided];
      a --[densely dotted] {[n=3,clockwise,nodes={shift={(a)}}]
        x[> solid,solution=undecided];
        y[coordinate];
        z[coordinate];
      };
    };
  \end{scope}
  \begin{scope}[xshift=5cm]
    \graph[grow right,branch down,nodes={every node},empty nodes] {
      a[propagating=dontcare,solution=undecided];
      a --[densely dotted] {[n=3,clockwise,nodes={shift={(a)}}]
        x[> solid,solution=no];
        y[coordinate];
        z[coordinate];
      };
    };
  \end{scope}
  \draw[snaked,->] (3cm,0) -- (4cm,0);
\end{tikzpicture}
    \caption{\ref{red:leaf 1}: If an undecided leaf is attached to a non-excluded vertex. The reduction rule excludes the leaf.\label{fig:red:leaf 1}}
  \end{subcaptionblock}
  \begin{subcaptionblock}{\linewidth}
    \centering
    \tikzpicturename{rules-leaf21}
\begin{tikzpicture}[pds instance,every node/.style={node}]
  \begin{scope}[xshift=2cm]
    \graph[grow right,branch down,nodes={every node},empty nodes] {
      a[propagating=yes,solution=dontcare];
      a --[densely dotted] {[n=3,clockwise,nodes={shift={(a)}}]
        x[> solid,solution=no];
        y[coordinate];
        z[coordinate];
      };
    };
  \end{scope}
  \begin{scope}[xshift=5cm]
    \graph[grow right,branch down,nodes={every node},empty nodes] {
      a[propagating=no,solution=dontcare];
      a --[densely dotted] {[n=3,clockwise,nodes={shift={(a)}}]
        x[> draw=none,coordinate];
        y[coordinate];
        z[coordinate];
      };
    };
  \end{scope}
  \draw[snaked,->] (3cm,0) -- (4cm,0);
\end{tikzpicture}
    \hspace{2cm}
    \tikzpicturename{rules-leaf22}
\begin{tikzpicture}[pds instance,every node/.style={node}]
  \begin{scope}[xshift=2cm]
    \graph[grow right,branch down,nodes={every node},empty nodes] {
      a[propagating=no,solution=undecided];
      a --[densely dotted] {[n=3,clockwise,nodes={shift={(a)}}]
        x[> solid,solution=no];
        y[coordinate];
        z[coordinate];
      };
    };
  \end{scope}
  \begin{scope}[xshift=5cm]
    \graph[grow right,branch down,nodes={every node},empty nodes] {
      a[propagating=no,solution=yes];
      a --[densely dotted] {[n=3,clockwise,nodes={shift={(a)}}]
        x[> draw=none,coordinate];
        y[coordinate];
        z[coordinate];
      };
    };
  \end{scope}
  \draw[snaked,->] (3cm,0) -- (4cm,0);
\end{tikzpicture}
    \caption{\ref{red:leaf 2}: Two cases of excluded leaves attached to a vertex. On the left, the parent vertex is propagating and becomes non-propagating. On the right, the parent vertex is non-propagating and becomes pre-selected. The attached leaf is removed in both cases.\label{fig:red:leaf 2}}
  \end{subcaptionblock}
  \begin{subcaptionblock}{\linewidth}
    \centering
    \tikzpicturename{rules-path1}
\begin{tikzpicture}[pds instance,every node/.style={node}]
  \begin{scope}[xshift=-1cm]
    \graph[grow right,nodes={every node},empty nodes] {
      / [coordinate,xshift=0.5cm] --[draw=none] a[propagating=dontcare,"$x$"] -- v[propagating=yes,"$v$"] -- b[propagating=dontcare,solution=dontcare,"$y$"];
      {[edges=densely dotted]
        a -- {[circular placement,phase=150,radius=0.5cm,nodes={coordinate,shift={(a)}}] subgraph I_n[n=2,name=a] };
        b -- {[circular placement,phase=330,radius=0.5cm,nodes={coordinate,shift={(b)}}] subgraph I_n[n=2,name=b] };
      };
    };
  \end{scope}
  \begin{scope}[xshift=5cm]
    \graph[grow right,nodes=every node,empty nodes] {
      a[propagating=dontcare] -- v[solution=no,propagating=yes] -- b[propagating=dontcare,solution=dontcare];
      {[edges=densely dotted]
        a -- {[circular placement,phase=150,radius=0.5cm,nodes={coordinate,shift={(a)}}] subgraph I_n[n=2,name=a] };
        b -- {[circular placement,phase=330,radius=0.5cm,nodes={coordinate,shift={(b)}}] subgraph I_n[n=2,name=b] };
      };
    };
  \end{scope}
  \draw[snaked,->] (3cm,0) -- (4cm,0);
\end{tikzpicture}
    \caption{\ref{red:path 1}: A vertex $v$ of degree two can safely be excluded if it has an undecided neighbor. \label{fig:red:path 1}}
  \end{subcaptionblock}
  \begin{subcaptionblock}{\linewidth}
    \centering
    \tikzpicturename{rules-path2}
\begin{tikzpicture}[pds instance,every node/.style={node}]
  \begin{scope}[xshift=-1cm]
    \graph[grow right,nodes={every node},empty nodes] {
      / [coordinate,xshift=0.5cm] -- a[propagating=yes,solution=dontcare,"$x$"] -- v[propagating=yes,solution=no,"$v$"] -- b[propagating=dontcare,solution=dontcare,"$y$"];
      a --[bend right,nonedge] b;
      {[edges=densely dotted]
        b -- {[circular placement,phase=330,radius=0.5cm,nodes={coordinate,shift={(b)}}] subgraph I_n[n=2,name=b] };
      };
    };
  \end{scope}
  \begin{scope}[xshift=4cm]
    \graph[grow right,nodes=every node,empty nodes] {
      / [coordinate,xshift=0.5cm] -- a1[propagating=yes,solution=dontcare] -- v[coordinate] -- b[propagating=dontcare,solution=dontcare];
      {[edges=densely dotted]
        b -- {[circular placement,phase=330,radius=0.5cm,nodes={coordinate,shift={(b)}}] subgraph I_n[n=2,name=b] };
      };
    };
  \end{scope}
  \draw[snaked,->] (b) (3cm,0) -- (4cm,0);
\end{tikzpicture}
    \caption{\ref{red:path 2}: The neighbors $x$ and $v$ of degree two can be merged if $x$ is not adjacent to $v$'s other neighbor~$y$. \label{fig:red:path 2}}
  \end{subcaptionblock}
  \begin{subcaptionblock}{\linewidth}
    \centering
    \tikzpicturename{rules-path4}
\begin{tikzpicture}[pds instance,every node/.style={node}]
  \begin{scope}
    \graph[grow right,nodes={every node},empty nodes] {
      / [coordinate,xshift=0.5cm] -- a[propagating=yes,solution=dontcare,"$x$"] -- z[propagating=yes,observed=yes,solution=no,"$v$"] -- v[propagating=yes,solution=no,"$y$"] -- b[propagating=dontcare,solution=dontcare,"$z$"] --[draw=none] / [xshift=-0.5cm,coordinate];
      a --[bend right,nonedge] b;
      {[edges=dotted]
        b -- {[circular placement,phase=330,radius=0.5cm,nodes={coordinate,shift={(b)}}] subgraph I_n[n=2,name=b] };
      };
    };
  \end{scope}
  \begin{scope}[xshift=6cm]
    \graph[grow right,nodes={every node},empty nodes] {
      / [coordinate,xshift=0.5cm] -- a[propagating=yes,solution=dontcare] -!- z[propagating=yes,observed=yes,solution=no] -!- v[coordinate] -!- b[propagating=dontcare,solution=dontcare] --[draw=none] / [xshift=-0.5cm,coordinate];
      a --[bend right] b;
      {[edges=dotted]
        b -- {[circular placement,phase=330,radius=0.5cm,nodes={coordinate,shift={(b)}}] subgraph I_n[n=2,name=b] };
      };
    };
  \end{scope}
  \draw[snaked,->] (5,0) -- (6,0);
\end{tikzpicture}
    \caption{\ref{red:path 3}: If $v$ is observed and has two non-adjacent unobserved neighbors of degree two, $x$ and $y$, with $y$ being excluded, then we remove $y$ and the edge $xv$ and connect $x$ to the other neighbor  $z$ of $y$.\label{fig:red:path 3}}
  \end{subcaptionblock}
  \begin{subcaptionblock}{\linewidth}
    \centering
    \tikzpicturename{rules-path3}
\begin{tikzpicture}[pds instance,every node/.style={node}]
  \begin{scope}[xshift=2cm]
    \graph[grow right,branch down,nodes={every node},empty nodes] {
      a[propagating=no,solution=undecided];
      {
        a -- {[circular placement,phase=60,nodes={propagating=dontcare,solution=dontcare,shift={(a)}}] subgraph P_n[n=2,name=a] };
        a --[densely dotted] {[circular placement,phase=240,radius=0.3cm,nodes={coordinate,shift={(a)}}] subgraph I_n[n=2,name=b] };
      };
    };
  \end{scope}
  \begin{scope}[xshift=5cm]
    \graph[grow right,branch down,nodes={every node},empty nodes] {
      a[propagating=no,solution=yes];
      {
        a --[densely dotted] {[circular placement,phase=240,radius=0.3cm,nodes={coordinate,shift={(a)}}] subgraph I_n[n=2,name=b] };
      };
    };
  \end{scope}
  \draw[snaked,->] (3cm,0) -- (4cm,0);
\end{tikzpicture}
    \caption{\ref{red:tri}: We pre-select the undecided and remove the two vertices of degree two in the triangle.\label{fig:red tri}}
  \end{subcaptionblock}
  \begin{subcaptionblock}{\linewidth}
    \centering
    \tikzpicturename{rules-path5}
\begin{tikzpicture}[pds instance,every node/.style={node}]
  \begin{scope}
    \graph[grow right,branch down,nodes={every node},empty nodes] {
      a[propagating=no,solution=no] -- v[propagating=dontcare,solution=no] -- b[propagating=no];
      {[edges=densely dotted]
        a -- {[circular placement,phase=150,radius=0.5cm,nodes={coordinate,shift={(a)}}] subgraph I_n[n=2,name=a] };
        b -- {[circular placement,phase=330,radius=0.5cm,nodes={coordinate,shift={(b)}}] subgraph I_n[n=2,name=b] };
      };
    };
  \end{scope}
  \begin{scope}[xshift=5cm]
    \graph[grow right,branch down,nodes={every node},empty nodes] {
      a1[propagating=no,solution=no] -- v[solution=no,propagating=dontcare,solution=no,observed] -- b1[propagating=no,solution=yes];
      {[edges=densely dotted]
        a1 -- {[circular placement,phase=150,radius=0.5cm,nodes={coordinate,shift={(a1)}}] subgraph I_n[n=2,name=a1] };
        b1 -- {[circular placement,phase=330,radius=0.5cm,nodes={coordinate,shift={(b1)}}] subgraph I_n[n=2,name=b1] };
      };
    };
  \end{scope}
  \draw[snaked,->] (b) (3cm,0) -- (4cm,0);
\end{tikzpicture}
    \caption{\ref{red:path 5}: If an excluded vertex is surrounded by non-propagating vertices, only one of which is not excluded, we pre-select the non-excluded vertex.\label{fig:red:path 5}}
  \end{subcaptionblock}
  \caption{Illustrated overview of the local reduction rules. Round vertices {\tikz[pds instance] \node[node,propagating] {};} are propagating, triangular vertices {\tikz[pds instance] \node[node,propagating=no] {};} are non-propagating, square vertices \tikz[pds instance] \node[node] {}; may be propagating or non-propagating. Hollow vertices {\tikz[pds instance] \node[node,solution=no] {};} are excluded from a solution, vertices filled black {\tikz[pds instance] \node[node,solution=undecided] {};} are undecided, red vertices {\tikz[pds instance] \node[node,solution=yes,propagating=no] {};} are pre-selected. Vertices filled gray {\tikz[pds instance] \node[node,solution=dontcare] {};} may be undecided, excluded or pre-selected. Green vertices {\tikz[pds instance] \node[node,observed,solution=no] {};} are observed but not pre-selected. The non-existence of an edge is indicated in red {\tikz[pds instance] \draw[nonedge] (0,0) -- ++(0.5,0);}.}
  \label{fig:local-red-rules}
\end{figure}

Our first set of reduction rules stems from the observation that leaves, i.e. vertices of degree one, are usually not part of an optimum solution.
To keep the rules safe, we only delete excluded leaves and only exclude leaves that can be observed from some other vertex.
See \cref{fig:red:leaf 1,fig:red:leaf 2} for a visual example.

\reductionname{Deg1a}
\begin{reduction}[restate=reductionDegIa]
  \label{red:leaf 1}
  Let $v$ be an undecided leaf attached to a non-excluded vertex $w$.
  Then mark $v$ as excluded.
\end{reduction}

\reductionname{Deg1b}
\begin{reduction}[restate=reductionDegIb]
  \label{red:leaf 2}
  Let $v$ be an excluded leaf attached to $w$.
  If $w$ is propagating, delete $v$ and set $w$ to non-propagating.
  If $w$ is not excluded and not propagating, delete $v$ and select $w$, i.e., set $X \gets X \cup \set{w}$.
\end{reduction}

A case very similar to \ref{red:leaf 2} is illustrated in \cref{fig:red tri}.
Two adjacent vertices $x$ and $y$ of degree two share an undecided neighbor $z$
Similar to \cref{lem:booster vertex}, we can always select $z$.

\reductionname{Tri} 
\begin{reduction}[restate=reductionTri]
  \label{red:tri}
  Let $x$ and $y$ be two adjacent vertices of degree two with a common neighbor $z$.
  If $z$ is undecided, select $z$ and delete $x$ and $y$.
\end{reduction}

The possibility of applying the propagation rule leads to a number of further reduction rules related two vertices of degree two, illustrated in \cref{fig:red:path 1,fig:red:path 2,fig:red:path 3}.
Note, for example, that vertices of degree two with an undecided neighbor never need to be selected.
One can always select the neighbor instead.

\reductionname{Deg2a}
\begin{reduction}[restate=reductionDegIIa]
  \label{red:path 1}
  Let $v \in Z \setminus (X \cup Y)$ be an undecided propagating vertex with $d(v) = 2$ and let $x$ and $y$ be its neighbors.
  If $x$ and $y$ are not adjacent and if at least one of $x$ and $y$ is not excluded, exclude $v$.
\end{reduction}

If two adjacent vertices have degree two, the propagation rule ensures that if one becomes observed, the other becomes observed, too.
The two vertices can thus be merged.

\reductionname{Deg2b}
\begin{reduction}[restate=reductionDegIIb]
  \label{red:path 2}
  Let $v \in Y \cap Z$ be an excluded propagating vertex with $d(v) = 2$ and let $x$ and $y$ be its neighbors.
  If $x$ and $y$ are not adjacent and if at least one of $x$ and $y$ is propagating and has degree 2, delete $v$ and add an edge between $x$ and $y$ instead.
\end{reduction}

Something similar happens when two unobserved vertices of degree two share an observed neighbor with no other unobserved neighbors.
Recall the gadget used in the removal of the booster edges in \cref{lem:booster gadget}, which works in the same way.
If one of the two unobserved vertices becomes observed, the other becomes observed by the propagation rule.
If both vertices additionally have degree two, they thus behave like a single vertex to their other two neighbors.

\reductionname{Deg2c}
\begin{reduction}[restate=reductionDegIIc]
  \label{red:path 3}
  Let $v$ be an observed excluded propagating vertex with precisely two unobserved neighbors, both of degree two, $x$ and $y$.
  If $x$ and $y$ are not connected, $y$ must have another neighbor $z$.
  In this case remove $y$ and the edge $xv$ and add a new edge $xz$.
\end{reduction}

There are several trivial cases where a vertex must be selected to form a feasible solution.
If an undecided unobserved vertex has no neighbors, clearly that vertex must be selected.
See \cref{fig:red:path 5} for an example.
Similarly, if an unobserved excluded vertex has only non-propagating only one of which is not excluded, that neighbor must be selected.

\reductionname{OnlyN} 
\begin{reduction}[restate=reductionOnlyN]
  \label{red:path 5}
  Let $v$ be an excluded propagating vertex of degree two with two non-propagating neighbors $x$ and $y$.
  If $x$ is excluded and $y$ is undecided, select $y$.
\end{reduction}

\reductionname{Isol}
\begin{reduction}[name=Isolated,restate=reductionIsol]
  \label{red:isolated}
  Let $v$ be an undecided isolated vertex.
  Then pre-select $v$.
\end{reduction}

\reductionname{ObsNP}
\begin{reduction}[name=Observed Non-Propagating,restate=reductionObsNP]
  \label{red:observed non-zi}
  Let $v$ be an observed, non-propagating and excluded vertex.
  Then delete $v$.
\end{reduction}
We do not need the propagation rule to observe the vertices observed by the currently selected vertices.
Instead we can connect all those vertices directly to selected vertices and observe them by the domination rule instead and remove the now redundant edges.
This rule is most useful when combined with the other rules.
In particular, all vertices with only observed neighbors become leaves and are then removed by \cref{red:leaf 1,red:leaf 2}.

\reductionname{ObsE}
\begin{reduction}[name=Observed Edge,restate=reductionObsE]
  \label{red:observed edge}
  Let $vw \in E$ be an edge with two observed but not pre-selected endpoints and let $x \in X$.
  Then delete $vw$ and instead insert edges $vx$ and $wx$.
\end{reduction}

The previous rules identified vertices to be forbidden from a solution by looking at small localized structures.
We can identify those vertices and many more by looking at the vertices that would become observed by selecting a vertex.
For the next two reduction rules, we introduce the concept of \emph{observation neighborhood}.
For a set of vertices $U \subseteq V$ the observation neighborhood is the set of vertices that is observed when selecting $U$ in addition to all pre-selected vertices and applying the observation rules exhaustively.
For convenience, we define the observation neighborhood of a single vertex $v$ to be the observation neighborhood of the single element set $\set{v}$.
We use the observation neighborhood to formulate the following reduction rule.

\reductionname{Dom}
\begin{reduction}[name=Domination,restate=reductionDom]
  \label{red:domination}
  If the closed neighborhood of some undecided vertex $w$ is contained in the observation neighborhood of some other undecided vertex $v$, exclude $w$.
\end{reduction}

\Cref{red:leaf 2,red:tri,red:path 5} can be seen as special cases of a more general rule.
All three identify vertices that must be active in a minimum solution.
There are, however some vertices that must be active in every solution.
We can identify all those vertices by looking at the observation neighborhood of all other undecided vertices.
This tells us if there is a solution in which the vertex is not selected.

\reductionname{NecN}
\begin{reduction}[name=Necessary Node,restate=reductionNecN]
  \label{red:necessary}
  Let $v \in B$ be an undecided vertex.
  If the observation neighborhood of all undecided vertices except $v$ does not contain all vertices in $G$, select $v$.
\end{reduction}

We note that Binkele-Raible and Fernau~\cite{binkele2012exact} already introduced reduction rules for their exponential-time algorithm.
We do not use them in our algorithm as they are not generally applicable but rather require a specific situation.
The only exception \cite[``isolated'']{binkele2012exact} that is applicable is superseded by a combination of our reduction rules.






\subsubsection{Order of Application}
\label{sec:order application}

In a first step, use a depth first search and process the vertices in post-order to apply the rules \labelcref{red:leaf 1,red:leaf 2,red:path 1}.
The order in which we process the vertices is important here, as it makes sure that attached paths are properly reduced.
This is only relevant for this first application of the reduction rules and in later applications, we process the vertices and edges in arbitrary order.
After this initial application, we iterate the following three steps until no reduction rules can be applied.
(i) Iterate the application of the local reduction rules (\ref{red:leaf 1}, \ref{red:leaf 2}, \ref{red:path 1}, \ref{red:path 2}, \ref{red:tri}, \ref{red:path 3}, \ref{red:path 5}, \ref{red:observed non-zi}, \ref{red:observed edge}) until no local reduction rule is applicable.
(ii) Apply the non-local rule \ref{red:domination}.
(iii) Apply the non-local rule \ref{red:necessary}.

We note that applying \ref{red:domination} once to all vertices is exhaustive in the sense that it cannot be applied again immediately afterwards.
It can, however, become applicable again after rerunning the other reduction rules.
The same is true for \ref{red:necessary}.

Our reasoning for this sequence of application is that the local rules are more efficient than the non-local once.
Thus, we first apply the cheap rules exhaustively before resorting to the expensive ones.
Preliminary experiments showed that further tweaking the order of application has only minor effect on the kernel size and run time.

\subsubsection{Implementation Notes}

The naïve implementation of the reduction rules can be very slow, in particular for the non-local rules.
The costly operation in those rules is the computation of the observation neighborhood.
We thus use a specialized data structure that allows us to pre-select and deselect vertices in arbitrary order.
Each time we pre-select a vertex, we also update the observed vertices and keep track of which vertex propagates to which other vertex.
For de-selecting vertices, we only mark vertices as unobserved, that were directly or indirectly observed by the deselected vertex.
Being able to select and deselect arbitrary vertices allows a straightforward implementation of the non-local rules.

\subsection{Split into Subinstances}
\label{sec:sub-problems}

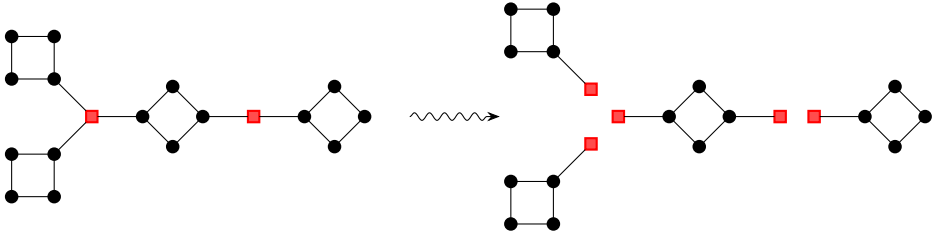
\begin{figure}
  \tikzpicturename{reduction-subproblems}
\begin{tikzpicture}[pds instance,node distance=1cm,graphs/radius=0.4cm,every node/.style={node},zi/.style={propagating=yes}]
\begin{scope}
  \node[solution] (a1) {};
  \coordinate[right=of a1] (graph);
  \graph[clockwise,empty nodes,nodes={zi,shift={(graph)}}] {
    subgraph C_n[V={b1,b2,b3,b4}];
  };
  \node[solution,right=of graph] (a2) {};
  \graph[use existing nodes] { a1 -- b4; b2 -- a2; };
  \coordinate[below left=of a1] (graph);
  \graph[clockwise,nodes={zi,shift={(graph)}},phase=45,,empty nodes]{
    subgraph C_n[V={c1,c2,c3,c4}];
  };
  \coordinate[above left=of a1] (graph);
  \graph[clockwise,nodes={zi,shift={(graph)}},phase=-45,empty nodes]{
    subgraph C_n[V={d1,d2,d3,d4}];
  };
  \graph[use existing nodes] { a1 -- {c1, d1};};
  \coordinate[right=of a2] (graph);
  \graph[clockwise,nodes={zi,shift={(graph)}},empty nodes] {
    subgraph C_n[V={e1,e2,e3,e4}];
  };
  \graph[use existing nodes] { a2 -- e4; };
\end{scope}

\coordinate[right=of graph] (mid left);

\begin{scope}[xshift=7cm]
  \node[solution] (a1) {};
  \coordinate[right=of a1] (graph);
  \graph[clockwise,empty nodes,nodes={zi,shift={(graph)}}] {
    subgraph C_n[V={b1,b2,b3,b4}];
  };
  \node[solution,right=of graph] (a2) {};
  \graph[use existing nodes] { a1 -- b4; b2 -- a2; };
  \node[solution,below left=0.3 of a1] (a11) {};
  \coordinate[below left=of a11] (graph);
  \graph[clockwise,nodes={zi,shift={(graph)}},phase=45,,empty nodes]{
    subgraph C_n[V={c1,c2,c3,c4}];
  };
  \node[solution,above left=0.3 of a1] (a12) {};
  \coordinate[above left=of a12] (graph);
  \graph[clockwise,nodes={zi,shift={(graph)}},phase=-45,empty nodes]{
    subgraph C_n[V={d1,d2,d3,d4}];
  };
  \graph[use existing nodes] { a11 -- c1; a12 -- d1;};
  \node[solution,right=0.3 of a2] (a21) {};
  \coordinate[right=of a21] (graph);
  \graph[clockwise,nodes={zi,shift={(graph)}},empty nodes] {
    subgraph C_n[V={e1,e2,e3,e4}];
  };
  \graph[use existing nodes] { a21 -- e4; };
\end{scope}

\coordinate[left=1.5 of a1] (mid right);

\draw[snaked,->] (mid left) -- (mid right);

\end{tikzpicture}
  \caption{Example of independent sub-problems arising from the kernel on the left. Each connected component on the right can be solved independent from the others.}
  \label{fig:red:sub problems}
\end{figure}

While the propagation rule may have non-local effects within the whole graph, propagation cannot pass through selected vertices.
This is formalized by the following theorem.

\begin{restatable}{theorem}{splitIntoSubinstances}
  Let $G = (V, E)$ be the graph with pre-selected vertices $X \subseteq V$ and let $C_1,\dots,C_\ell \subseteq V$ be the vertices in the connected components of the sub-graph of $G$ induced by $V \setminus X$.
  Further let $S_1,\dots,S_\ell$ be minimum power dominating sets of the subgraphs induced by $N[C_1],\dots,N[C_\ell]$.
  Then $S = S_1 \cup \dots \cup S_\ell$ is a minimum power dominating set of $G$.
\end{restatable}
\begin{proof}
  Let $G_i$ be the induced sub-graph on $N[S_i]$ and let $S_i$ be a minimum solution of $G_i$.
  Let $S = S_1 \cup \dots \cup S_\ell$ be a candidate solution.
  We show that each $S_i$ observes at least the vertices $C_i$ in $G$.
  Now consider a sequence of observation rules observing all vertices in $G_i$.
  Clearly, the domination rule can be applied in $G$ in the same way as in $G_i$.
  No unselected vertex in $G_i$ has neighbors outside $G_i$, so the propagation rule can also be applied in the same way.
  As this applies to all $G_i$, all vertices in $G$ are observed and thus $S$ is a solution.

  It remains to show that $S$ is a minimum solution.
  Assume $S$ is not a minimum power dominating set and let $S*$ be a minimum solution.
  Then there is a sub-graph $G_i$ with $\abs{C_i \cap S_I} > \abs{C_i \cap S*}$ and $C_i \cap S*$ must be a power dominating set of $G_i$.
  This contradicts the assumption that $S_i$ is minimal.
\end{proof}

These sub-problems can be identified in linear time using a depth-first search restarting at unexplored non-active nodes while ignoring outgoing edges of active nodes.
See \cref{fig:red:sub problems} for an example of an instance split into subinstances.

\subsection{Solving the Kernel Via Implicit Hitting Set}
\label{sec:kernel solver}

We briefly describe the implicit hitting set approach.
Compared to how Bozeman et al.~\cite{bozeman2018RestrictedPD} introduced it, we allow non-propagating vertices.
However, this does not change any of the proofs and thus the approach directly translates to this slightly more general setting.

In a graph $G$, a \emph{fort} is a non-empty subset of vertices $F \subseteq V(G)$ such that no propagating vertex outside $F$ is adjacent to precisely one vertex in $F$.
A power dominating set must be a hitting set of the family of all \emph{fort neighborhoods} in $G$, i.e., if $F$ is a fort and $S$ is a power dominating set, then $N[F] \cap S \neq \emptyset$.
Conversely, if a hitting set for a family $\mathcal F$ of fort neighborhoods is not a power dominating set, then one can find an additional fort of $G$ whose neighborhood is not in $\mathcal F$.

This yields the following algorithm.
Start with some set $\mathcal F$ of fort neighborhoods.
Compute a minimum hitting set $H$ for $\mathcal F$.
If $H$ is already a power dominating set, we have found the optimum.
Otherwise, we construct at least one new fort neighborhood and add it to $\mathcal F$.

One core ingredient of this approach is the choice of which fort neighborhoods to add to~$\mathcal F$.
Previous approaches~\cite{smith2020OptimalSP,bozeman2018RestrictedPD} aimed at finding forts or fort neighborhoods that are as small as possible.
The reasoning behind this is that the set of all fort neighborhoods can be exponentially large (even when restricted to those that are minimal with respect to inclusion) and thus it makes sense to add sets that are as restrictive as possible, hoping that only few sets suffice before the \HS solution yields a \Pds solution.
However, finding forts of minimum size or minimum size fort neighborhoods is difficult while just finding any fort is easy.
Moreover, if we add only few forts in every step, we have to potentially solve more \HS instances.
We thus propose to instead find multiple forts at once and to add them all to the \HS instance.
Our method of finding forts is based on the following lemma.

\begin{lemma}
  \label{lem:unobserved fort}
  Let $G = (V, E)$ be a graph and let $S \subseteq V$ be a set of selected vertices.
  Let further be $R$ the set of vertices observed by exhaustive application of the observation rules with respect to $S$.
  Then the set of unobserved vertices $V \setminus R$ is a fort.
\end{lemma}
\begin{proof}
  Assume $V \setminus R$ is not a fort.
  Then there exists a propagating vertex $v$ in $R$ that is adjacent to precisely one vertex $w$ in $V \setminus R$, i.e., $v$ has precisely one unobserved neighbor $w$.
  This constradicts the exhaustive application of the propagation rule and thus the set of unobserved vertices is a fort.
\end{proof}

By Lemma~\ref{lem:unobserved fort}, whenever we have a candidate solution that does not yet observe all vertices, we obtain a new fort and can add its neighborhood to the \HS instance $\mathcal F$.
For the new forts, we have two objectives.
First, we want the new fort neighborhood to actually provide new restrictions, i.e., it should not be already hit by the minimum hitting set $H$ of $\mathcal F$.
This is achieved by making sure that the candidate solution $S$ is a superset of the hitting set $H$.
Secondly, we want the resulting forts (i.e., the number of unobserved vertices) to be small.
We achieve this heuristically by greedily considering large candidate solutions.

Specifically, we choose candidate solutions as follows.
Recall, that we consider the extension problem, i.e., we have sets $X$ and $Y$ of pre-selected and excluded vertices, respectively.
Moreover, let $H$ be a minimum hitting set of the current set of fort neighborhoods.
Then $V$ is partitioned into the four sets $X$, $Y$, $H$, and $U = V \setminus H \setminus X \setminus Y$.
Each candidate solution $S$ we consider is a superset of $H \cup X$ and a subset of $H \cup X \cup U$.
We randomly order the vertices in $U = \set{u_1, \dots, u_\ell}$ and define a sequence $U_0, \dots, U_\ell \subseteq U$.
We then consider the candidate solutions $S_i = H \cup X \cup U_i$ for $0 \le i \le \ell$.
As we want to consider large candidate solutions, we start with $U_0 = U$, which clearly yields a solution as the instance would be invalid otherwise.
We obtain the subset $U_i$ from $U_{i - 1}$ as follows.
If $S_{i - 1}$ was a solution, i.e., there were no unobserved vertices, then $U_i = U_{i - 1} \setminus \set{u_i}$.
Otherwise, $U_i = U_{i - 1} \cup \set{u_{i - 1}} \setminus \set{u_i}$.
Note that this makes sure that each candidate solution $S_i$ we consider is either a solution or barley not a solution as $S_i \cup \set{u_i}$ is a solution.

This gives us at least one and up to $\ell$ new fort neighborhoods.
These are not directly added to the set $\mathcal F$.
Instead, we first apply a simple local search to make sure that each fort is minimal with respect to inclusion.
To this end, we iteratively re-select vertices from $U$ that had been removed before and check whether this still results in a non-empty fort.

We note that we only add sets to $\mathcal F$.
Thus, we have to solve a sequence of increasing \HS instances as a subroutine.
To improve the performance of this, one can use lower bounds achieved in earlier iterations as lower bounds for later iterations (\HS is monotone with respect to the addition of sets).







  \section{Experiments}
\label{sec:experiments}

The goal of this section is threefold.
First, we evaluate the performance of our algorithm in comparison to two previous state-of-the-art approaches.
Secondly, we give a more detailed view on the performance by analyzing how the upper and lower bounds found be the different algorithms converge to the optimal solution.
Thirdly, we evaluate the impact of the different reduction rules.

\subparagraph{Experiment Setup.}

We implemented our algorithm in C++ 20 and compiled it with clang 15.0.1 with the \texttt{-O3} optimization flag.
Our source code along with all data sets and evaluation scripts is available on GitLab\footnote{\url{https://gitlab.com/Aldorn/pds-code}}.
For the comparison with the previous state-of-the-art, we use the MILP formulation approach by Jovanovic and Voss \cite{jovanovic2020fixed}.
In the following, we refer to this algorithm with \texttt{MILP}.
The second solver by Smith and Hicks~\cite{smith2020OptimalSP} and is based on the implicit hitting set approach.
Unfortunately, their code is not publicly available, and the paper does not specify all implementation details.
To make a fair (or rather generous) comparison, we initialized their set of forts with our fort heuristic, which, as far as we can judge, leads to better results than reported in the original publication~\cite{smith2020OptimalSP}.
We refer to this algorithm as \texttt{MFN} (abbreviation for \emph{minimum fort neighborhood}).
For the implicit hitting set approaches, we use an MILP formulation to solve the \HS instances.
All MILP instances are solved using Gurobi~9.5.2~\cite{gurobi}.

The experiments were run on a machine running Ubuntu 22.04 with Linux 5.15.
The machine has two Intel\textregistered Xeon\textregistered Gold 6144 CPUs clocked at 3.5 GHz with 8 single-thread cores and 192 GB of RAM.

We used a collection of instances shipped with pandapower~\cite{thurner2018pandapower}.
We further use the Eastern, Western, Texas and US instances from the powersimdata set\footnote{\url{https://github.com/Breakthrough-Energy/PowerSimData}}~\cite{Xu2020powersimdata} based on the US electric grids.
We interpret the power grids as graphs were buses are vertices and power lines and transformers are edges.
Buses without attached loads or generators yield propagating vertices.
For experiments on the pandapower instances, we used a timeout of \SI{2}{h} and repeated each experiment 5 times.
On the powersimdata instances, we used a timeout of \SI{10}{h} and only repeated the experiments using our solver.
For repeated experiments, we report the median result.

\subparagraph{Performance Comparison.}

We compare the performance of our solver to the \texttt{MILP} and \texttt{MFN} approach, each with and without preprocessing by the reduction rules.
To assess the performance of our approach with reduction rules, we compute the speedup compared to the lowest run time of the previous approaches without reduction rules.

\begin{table}
\caption{Run times of different combinations of \Pds solvers and reduction rules on the pandapower data set. Note that $\gamma_P$ differs from the results reported in other literature. This is to be expected because we include non-propagating vertices from the input. Further observe that some run times are given in milli- or microseconds.}
\label{tbl:solver comparison normal}
\footnotesize
  \begin{threeparttable}
\begin{filecontents*}{tables/solver-times.csv}
name,n,blank_kern,pmus,lower,upper,bozeman-smith,bozeman-smith-r,jovanovic,jovanovic-r,ours,ours-r,speedup
4gs,4,4,2,2,2,{1.60m},{1.44m},{862µ},{1.09m},{\textbf{496µ}},{518µ},1.664092664092664
5,5,4,2,2,2,{1.08m},{1.34m},{1.07m},{676µ},{290µ},{\textbf{285µ}},3.7403508771929825
6ww\tnote{c},6,0,1,1,1,{450µ},{20µ},{1.07m},{\textbf{6µ}},{188µ},{\textbf{6µ}},75.0
9,9,6,2,2,2,{3.77m},{2.23m},{1.89m},{900µ},{708µ},{\textbf{484µ}},3.9132231404958677
11\_iwamoto\tnote{c},11,0,2,2,2,{3.54m},{\textbf{18µ}},{3.51m},{\textbf{18µ}},{546µ},{\textbf{18µ}},194.77777777777777
14,14,9,3,3,3,{2.67m},{2.24m},{1.63m},{1.28m},{530µ},{\textbf{460µ}},3.55
24\_ieee\_rts,24,16,6,6,6,{5.40m},{3.76m},{4.27m},{4.97m},{1.24m},{\textbf{726µ}},5.880165289256198
30\tnote{c},30,0,6,6,6,{4.69m},{143µ},{2.83m},{58µ},{704µ},{\textbf{56µ}},50.535714285714285
ieee30\tnote{c},30,0,6,6,6,{4.35m},{\textbf{43µ}},{3.41m},{132µ},{602µ},{44µ},77.54545454545455
33bw\tnote{c},33,0,11,11,11,{5.75m},{54µ},{1.82m},{\textbf{45µ}},{712µ},{144µ},12.625
39\tnote{c},39,0,9,9,9,{12.73m},{\textbf{103µ}},{6.16m},{267µ},{1.76m},{120µ},51.358333333333334
57,57,35,12,12,12,{23.58m},{15.34m},{8.91m},{9.93m},{\textbf{1.57m}},{1.79m},4.965459610027855
89pegase\tnote{c},89,0,13,13,13,{11.90m},{423µ},{21.57m},{171µ},{2.46m},{\textbf{169µ}},70.42011834319527
118,118,72,29,29,29,{64.71m},{51.13m},{13.72m},{12.59m},{\textbf{5.03m}},{5.90m},2.3239627434377645
145\tnote{c},145,0,18,18,18,{28.94m},{271µ},{121.84m},{\textbf{267µ}},{4.19m},{269µ},107.59107806691449
illinois200\tnote{c},200,0,39,39,39,{166.20m},{278µ},{20.38m},{\textbf{273µ}},{25.27m},{283µ},72.02120141342756
300,300,12,72,72,72,{204.81m},{4.19m},{27.94m},{2.80m},{11.74m},{\textbf{1.40m}},20.002147458840373
1354pegase,1354,3,311,311,311,2.06,{2.85m},{101.02m},{2.50m},{40.89m},{\textbf{1.95m}},51.88546481766821
1888rte,1888,6,375,375,375,6.34,{5.20m},{552.09m},{6.54m},{125.95m},{\textbf{3.07m}},179.5982433311646
2848rte,2848,6,585,585,585,15.37,{6.43m},{602.79m},{7.41m},{154.96m},{\textbf{5.04m}},119.52984334721395
2869pegase,2869,502,612,612,612,16.43,1.59,1.22,{220.63m},{182.89m},{\textbf{73.88m}},16.52525580639922
3120sp,3120,384,768,768,768,37.89,1.68,1.34,{273.92m},{252.51m},{\textbf{97.74m}},13.744193660602836
6470rte,6470,112,1303,1303,1303,89.97,{71.98m},2.99,{44.27m},{720.30m},{\textbf{26.58m}},112.50869009103904
6495rte,6495,135,1314,1314,1314,87.05,{92.95m},3.56,{46.35m},{825.62m},{\textbf{27.17m}},130.94722314232087
6515rte,6515,135,1315,1315,1315,92.09,{97.57m},3.93,{45.93m},{800.46m},{\textbf{27.75m}},141.72270663688116
9241pegase,9241,1845,2010,2010,2010,217.46,13.67,5.70,1.32,{820.84m},{\textbf{520.57m}},10.953022738492924
\end{filecontents*}
\pgfplotstabletypeset[
  string type,
  columns={name,pmus,jovanovic,jovanovic-r,bozeman-smith,bozeman-smith-r,ours,ours-r,speedup},
  columns/name/.style={
    column type=l,
    column name=instance,
    string type,postproc cell content/.append style={/pgfplots/table/@cell content/.add={\ttfamily}{}}
  },
  columns/pmus/.style={column name={$\gamma_P$}},
  columns/bozeman-smith/.style={column name={\texttt{MFN}\tnote{a}}},
  columns/bozeman-smith-r/.style={column name={\texttt{MFN+R}\tnote{a}}},
  columns/jovanovic/.style={column name={\texttt{MILP}\tnote{a}}},
  columns/jovanovic-r/.style={column name={\texttt{MILP+R}\tnote{a}}},
  columns/ours/.style={column name={Ours}},
  columns/ours-r/.style={column name={Ours+R}},
  columns/speedup/.style={column name={Speedup\tnote{b}},numeric type,fixed, fixed zerofill,precision=1},
  every head row/.style={
    before row=\toprule,
    after row={\texttt{case*}&\multicolumn{1}{c}{\#}&\multicolumn{1}{c}{s}&\multicolumn{1}{c}{s}&\multicolumn{1}{c}{s}&\multicolumn{1}{c}{s}&\multicolumn{1}{c}{s}&\multicolumn{1}{c}{s}\\\midrule},
  },
  every last row/.style={
    after row=\bottomrule,
  },
  every even row/.style={before row={\rowcolor[gray]{0.9}}},
  column type=r,multicolumn names,
]{tables/solver-times.csv}
\begin{tablenotes}
  \item[a] numbers here were obtained from our interpretation of the respective approach
  \item[b] speedup of Ours+R compared to the faster of \texttt{MILP} and \texttt{MFN}
  \item[c] solved optimally by reduction rules
\end{tablenotes}
\end{threeparttable}
\end{table}

\Cref{tbl:solver comparison normal} shows the run times of the solvers on the smaller pandapower instances.
Preprocessing significantly reduced the running times of all solvers in most cases, especially for the larger instances.
In fact, we found that our reduction rules were able to solve 9 out of 26 instances on their own.
In those cases, no solver without reduction rules could compete.
Out of the remaining instances, our our solver without reduction rules was the fastest on 3 instances while our solver with reduction rules was the fastest on all others.

For the larger powersimdata instances, neither \texttt{MILP} nor \texttt{MFN} were able to compute an optimal solution without using our reduction rules within the time limit.
Thus, for these instances, we only compare our solver with \texttt{MFN+R} and \texttt{MILP+R}.
Table~\ref{tbl:solver comparison large} shows the results.
Observe that for \texttt{Eastern}, our algorithm finished after \SI{16}{min} while \texttt{MFN} did not finish after more than \SI{6}{h}, with a lower bound that was still more than 100 vertices below the optimal solution.
Observe that the number of fort neighborhoods $|\mathcal F|$ is slightly lower for \texttt{MFN}, which is to be expected as this is basically the main goal of \texttt{MFN} when finding new fort neighborhoods.
However, this clearly does not show any benefit in the resulting run time.

\begin{table}
  \caption{Run times of different combinations of \Pds solvers and reduction rules on the pandapower data set with all vertices considered propagating. Observe that some run times are given in milli- or microseconds.}
  \label{tbl:solver comparison all propagating}
\footnotesize
  \begin{threeparttable}
    \begin{filecontents*}{tables/solver-times-z.csv}
name,n,blank_kern,pmus,lower,upper,bozeman-smith,bozeman-smith-r,jovanovic,jovanovic-r,ours,ours-r,speedup
4gs\tnote{c},4,0,1,1,1,{388µ},{3µ},{2.40m},{\textbf{2µ}},{196µ},{3µ},129.33333333333334
5\tnote{c},5,0,1,1,1,{570µ},{11µ},{2.99m},{\textbf{7µ}},{241µ},{19µ},30.0
6ww\tnote{c},6,0,1,1,1,{1.53m},{12µ},{3.30m},{12µ},{251µ},{\textbf{7µ}},218.42857142857142
9,9,3,2,2,2,{2.73m},{1.55m},{3.47m},{3.24m},{616µ},{\textbf{249µ}},10.951807228915662
11\_iwamoto\tnote{c},11,0,2,2,2,{980µ},{\textbf{8µ}},{3.73m},{10µ},{295µ},{\textbf{8µ}},122.5
14\tnote{c},14,0,2,2,2,{3.67m},{27µ},{15.35m},{57µ},{282µ},{\textbf{22µ}},166.86363636363637
24\_ieee\_rts\tnote{c},24,0,3,3,3,{5.34m},{\textbf{48µ}},{166.49m},{61µ},{590µ},{60µ},88.98333333333333
30\tnote{c},30,0,3,3,3,{6.70m},{\textbf{33µ}},{28.52m},{44µ},{1.29m},{37µ},181.02702702702703
ieee30\tnote{c},30,0,3,3,3,{8.26m},{46µ},{23.13m},{48µ},{1.11m},{\textbf{44µ}},187.6590909090909
33bw\tnote{c},33,0,2,2,2,{10.22m},{12µ},{15.17m},{\textbf{10µ}},{968µ},{13µ},786.0769230769231
39\tnote{c},39,0,5,5,5,{7.30m},{54µ},{24.43m},{\textbf{53µ}},{1.60m},{66µ},110.66666666666667
57,57,15,3,3,3,{18.06m},{11.79m},{154.01m},{55.12m},{4.52m},{\textbf{924µ}},19.541125541125542
89pegase\tnote{c},89,0,5,5,5,{12.88m},{364µ},1.26,{433µ},{39.39m},{\textbf{165µ}},78.04242424242425
118,118,33,8,8,8,{78.32m},{49.09m},{498.66m},{105.68m},{11.31m},{\textbf{5.90m}},13.276996101034074
145\tnote{c},145,0,13,13,13,{23.13m},{\textbf{332µ}},25.14,{354µ},{4.62m},{336µ},68.85119047619048
illinois200\tnote{c},200,0,20,20,20,{15.32m},{359µ},{46.67m},{\textbf{116µ}},{459.57m},{\textbf{116µ}},132.10344827586206
300,300,41,30,30,30,{312.35m},{97.38m},6.16,{372.88m},{272.78m},{\textbf{17.12m}},18.24160485896163
1354pegase,1354,18,176,176,176,1.27,{15.71m},1.52,{20.64m},{513.53m},{\textbf{3.80m}},333.6742105263158
1888rte,1888,9,235,235,235,1.79,{4.50m},5.49,{7.99m},{554.25m},{\textbf{2.58m}},693.6074476338247
2848rte,2848,7,352,352,352,2.15,{5.21m},9.42,{4.52m},{732.33m},{\textbf{3.29m}},653.4536333231986
2869pegase,2869,114,305,305,305,8.40,{326.75m},469.69,2.00,8.07,{\textbf{44.05m}},190.68121864287497
3120sp,3120,404,231,231,231,30.78,3.49,659.27,5.04,471.23,{\textbf{120.08m}},256.32637153992204
6470rte,6470,37,745,745,745,16.26,{34.71m},52.80,{169.64m},6.74,{\textbf{14.84m}},1095.5825471698113
6495rte,6495,36,757,757,757,14.62,{37.40m},647.15,{161.39m},5.04,{\textbf{15.93m}},917.4265988828218
6515rte,6515,36,757,757,757,14.74,{36.92m},85.10,{144.25m},5.77,{\textbf{14.88m}},990.3574067853543
9241pegase,9241,605,811,811,811,218.88,4.65,"{>7,200.00}",43.49,37.73,{\textbf{481.43m}},454.6547625724143
\end{filecontents*}
\pgfplotstabletypeset[
    ignore chars={"},
    string type,
  columns={name,pmus,jovanovic,jovanovic-r,bozeman-smith,bozeman-smith-r,ours,ours-r,speedup},
    columns/name/.style={
      column type=l,
      column name=instance,
      string type,postproc cell content/.append style={/pgfplots/table/@cell content/.add={\ttfamily}{}}
    },
    columns/pmus/.style={column name={$\gamma_P$}},
    columns/bozeman-smith/.style={column name={\texttt{MFN}\tnote{a}}},
    columns/bozeman-smith-r/.style={column name={\texttt{MFN+R}\tnote{a}}},
    columns/jovanovic/.style={column name={\texttt{MILP}\tnote{a}}},
    columns/jovanovic-r/.style={column name={\texttt{MILP+R}\tnote{a}}},
    columns/ours/.style={column name={Ours}},
    columns/ours-r/.style={column name={Ours+R}},
    columns/speedup/.style={column name={{Speedup}\tnote{b}},numeric type,fixed, fixed zerofill,precision=1},
    every head row/.style={
      before row=\toprule,
      after row={\texttt{case*}&\multicolumn{1}{c}{\#}&\multicolumn{1}{c}{s}&\multicolumn{1}{c}{s}&\multicolumn{1}{c}{s}&\multicolumn{1}{c}{s}&\multicolumn{1}{c}{s}&\multicolumn{1}{c}{s}\\\midrule},
    },
    every last row/.style={
      after row=\bottomrule,
    },
    every even row/.style={before row={\rowcolor[gray]{0.9}}},
    column type=r,multicolumn names,
    ]{tables/solver-times-z.csv}
    \begin{tablenotes}
      \item[a] numbers here were obtained from our interpretation of the respective approach
      \item[b] speedup of Ours+R compared to the faster of \texttt{MILP} and \texttt{MFN}
      \item[c] solved optimally by reduction rules
    \end{tablenotes}
  \end{threeparttable}
\end{table}

In the literature, most other solvers only consider networks consisting solely of propagating vertices.
For comparison, we conducted the experiment on the same instances but with all vertices propagating, see \cref{tbl:solver comparison all propagating}.
Observe that $\gamma_P$ is lower when all vertices are propagating.
With only propagating vertices in the input, our reduction rules could solve 13 out of 26 instances on their own.
On all remaining instances, our solver combined with the reduction rules was the fastest and achieved a median speedup of 176.2.

Comparing the results in \cref{tbl:solver comparison all propagating} and \cref{tbl:solver comparison normal}, we observe that the solvers react differently to all vertices being propagating.
While \texttt{MILP} and our solver without reductions are somewhat faster with some non-propagating vertices in the input, \texttt{MFN} and our solver with reductions perform better when all vertices are propagating.
This leads to an interesting effect: Overall, \texttt{MILP} performs better than \texttt{MFN} when the input contains non-propagating vertices but considerably  worse when all vertices are propagating.
Nonetheless, our solver remains the fastest in both cases.


\begin{table}
  \footnotesize
  \caption{
    Comparison between our algorithm, \texttt{MILP} and \texttt{MFN} on the larger powersimdata US instances preprocessed with our reduction rules.
    The number of vertices is $n$, $|Z|$ is the number of non-propagating vertices and $\abs{\mathcal F}$ is the size of the arising hitting set instance, if applicable.
    For the solvers, we report the power dominating number $\gamma_P$ (or the best found lower bound) as well as the number of fort neighborhoods $\mathcal F$ and the run time.
  }
  \label{tbl:solver comparison large}
  \begin{tabular}{l r r r r r r r r r r}
    \toprule
  & \multicolumn{2}{c}{Input} & \multicolumn{3}{c}{Our  Solver} & \multicolumn{3}{c}{\texttt{MFN+R}} & \multicolumn{2}{c}{\texttt{MILP+R}} \\\cmidrule(rl){2-3}\cmidrule(rl){4-6}\cmidrule(rl){7-9}\cmidrule(rl){10-11}
  Instance                             & $n$   & $\abs{Z}$ & $\gamma_P$ & $\abs{\mathcal{F}}$ & $t$~(s) & $\gamma_P$ & $\abs{\mathcal{F}}$ & $t$~(s) & $\gamma_P$ & $t$~(s) \\\midrule
  \rowcolor[gray]{0.9}\texttt{Texas}   & 2000  & 376   & 411   & 838   & 0.98   & 411    & 659     & 17.73        & 411    & 1.81 \\
  \texttt{Western}                     & 10024 & 4106  & 1825  & 2618  & 1.55   & 1825   & 2010    & 158.51       & 1825   & 2.16 \\
  \rowcolor[gray]{0.9}\texttt{Eastern} & 70047 & 30332 & 12895 & 27019 & 552.46 & >12789 & >15043  & >\SI{10}{h}  & >12890 & >\SI{10}{h} \\
  \texttt{USA}                         & 82071 & 34814 & 15131 & 30357 & 728.62 & >14124 & >16391  & >\SI{10}{h}  & >15126 & >\SI{10}{h} \\\bottomrule
  \end{tabular}
\end{table}

\subparagraph{Lower and Upper Bounds.}

We note that all three approaches find lower bounds while solving the instances.
In case of the implicit hitting set approeach, each time we solve the current \HS instance, the solution size is a lower bound for a minimum power dominating set.
This yields lower bounds for our approach as well as for \texttt{MFN}.
Moreover, Gurobi also provides lower bounds for \texttt{MILP}.
Additionally, Gurobi provides upper bounds.
To also get upper bounds for the implicit hitting set approaches, we use the following greedy heuristic.
Whenever we have computed a hitting set $H$ of the current fort neighborhoods, we greedily add vertices to $H$, preferably selecting vertices with many unobserved neighbors, until we have a power dominating set.
Afterwards, we make sure that the resulting solution is minimal with respect to inclusion.

With this, we can observe how quickly the different algorithms converge towards the optimal solution.
\Cref{fig:exp:bounds and gap} illustrates the behavior of the bounds with respect to the time for two of the four powersimdata instances.
All three algorithms use our reduction rules (recall that neither \texttt{MILP} nor \texttt{MFN} were able to solve these instances without them).
We clearly see that, with our approach, the gap between upper and lower bounds shrinks quickly, in particular compared to \texttt{MFN}.
This validates our assumption that adding many -- potentially larger -- forts instead of a single minimum size one is highly beneficial.
Recall that \texttt{MFN} can increase its lower bound only by at most $1$ after finding a new hitting set while we can increase the lower by up to one for each undecided unhit vertex.

Interestingly, for \texttt{MILP+R} the gap between upper and lower bound closes much quicker than for \texttt{MFN+R}.
In particular, for the largest \texttt{USA} instance, there is almost no gap left after little more than \SI{100}{s}.
Gurobi also found an optimal solution, but failed to prove the lower bound on its size within the timeout of \SI{10}{h}.
Thus, in cases where a good approximation is acceptable, the \texttt{MILP} formulation (with our reduction rules) is not much worse than our approach.

\begin{figure}
  \begin{center}
    \tikzpicturename{evaluation-bounds_gap}
\tikzpicturedependsonfile{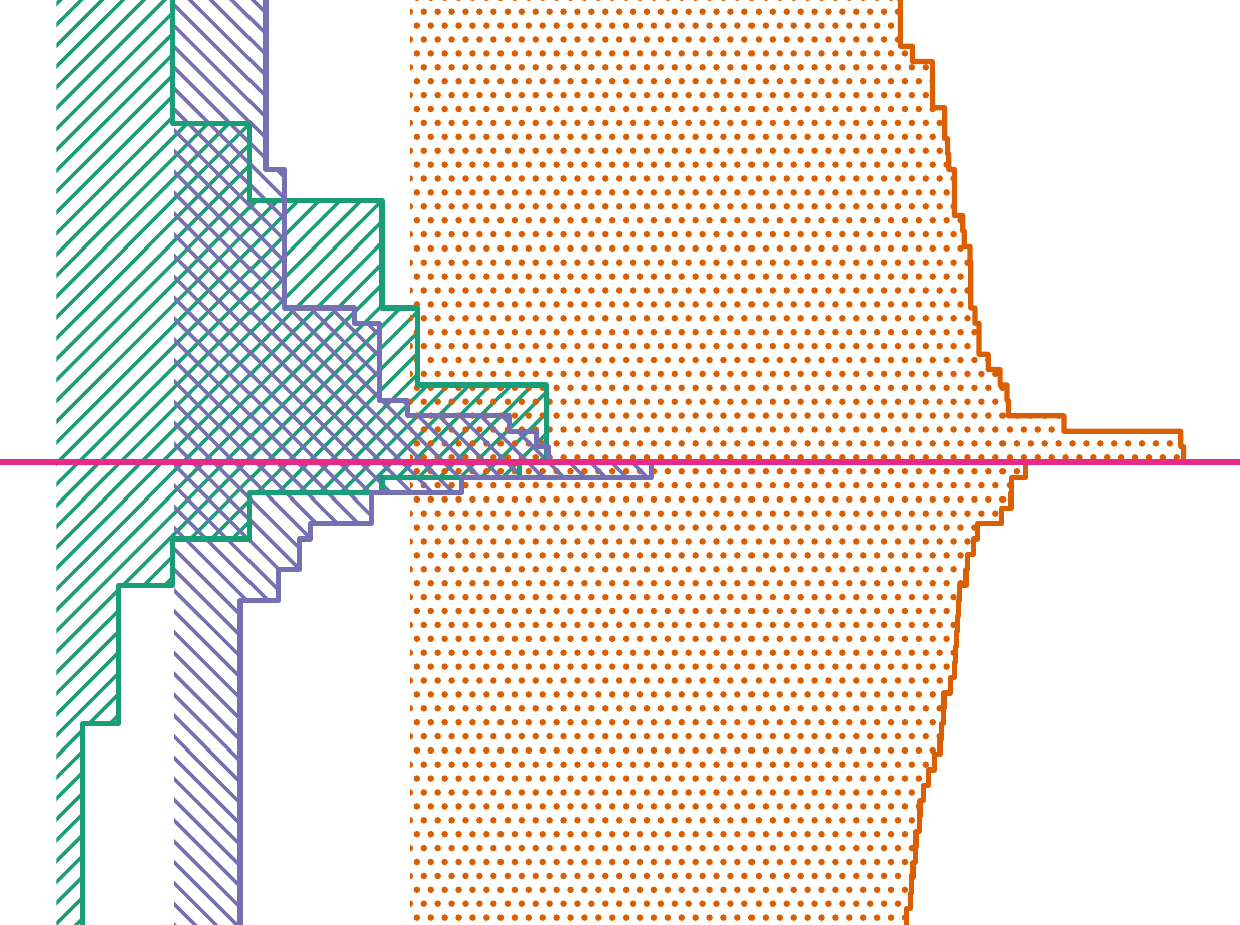}
\tikzpicturedependsonfile{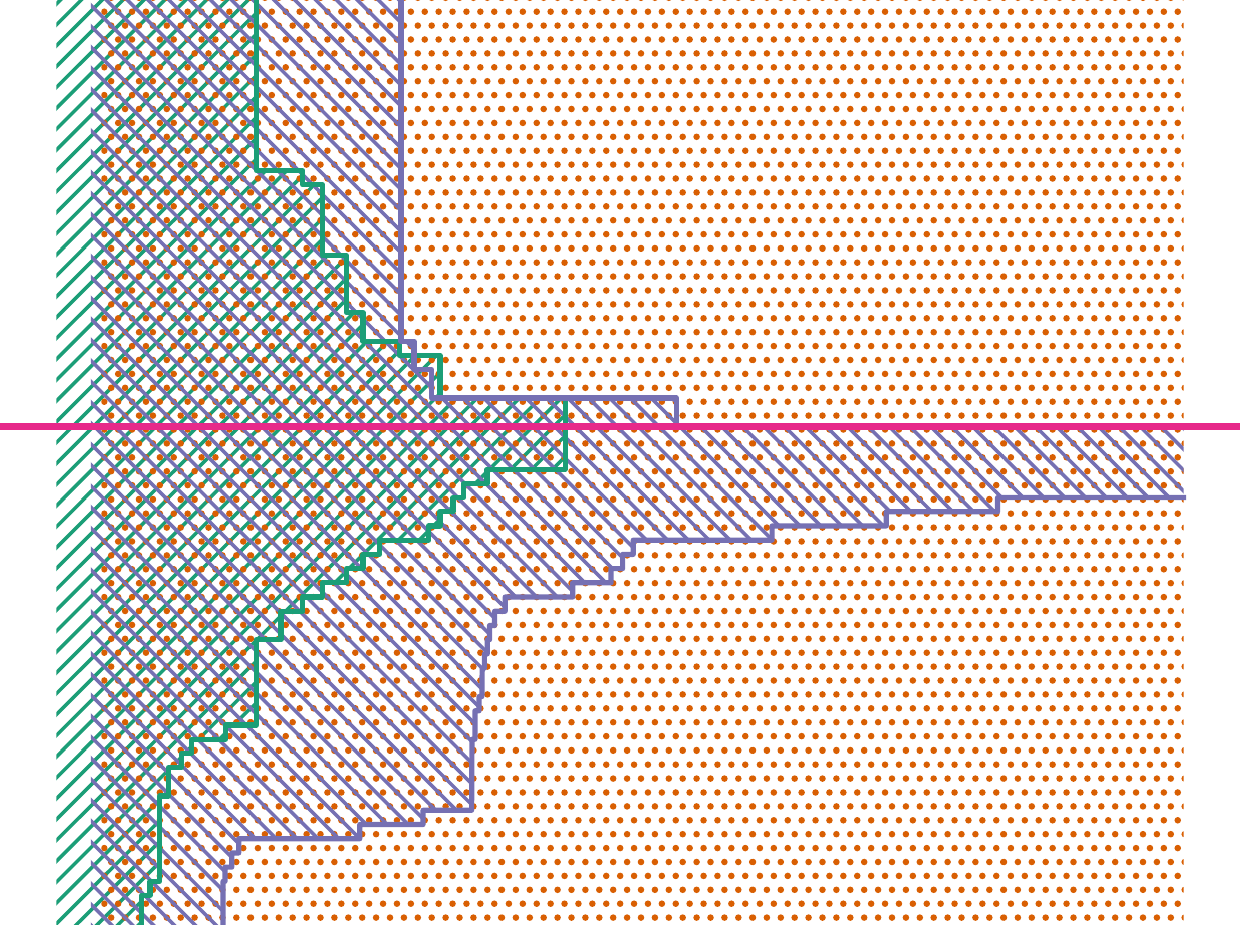}
\tikzpicturedependsonfile{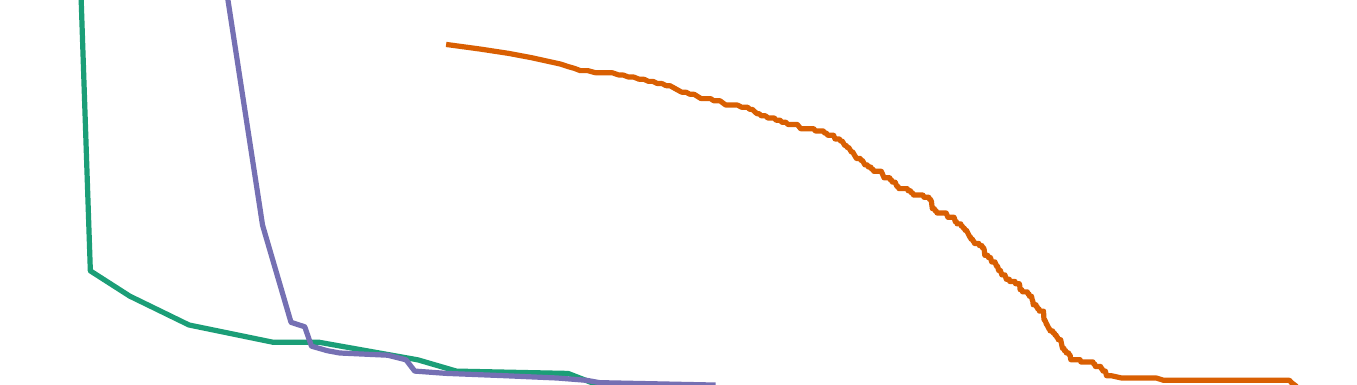}
\tikzpicturedependsonfile{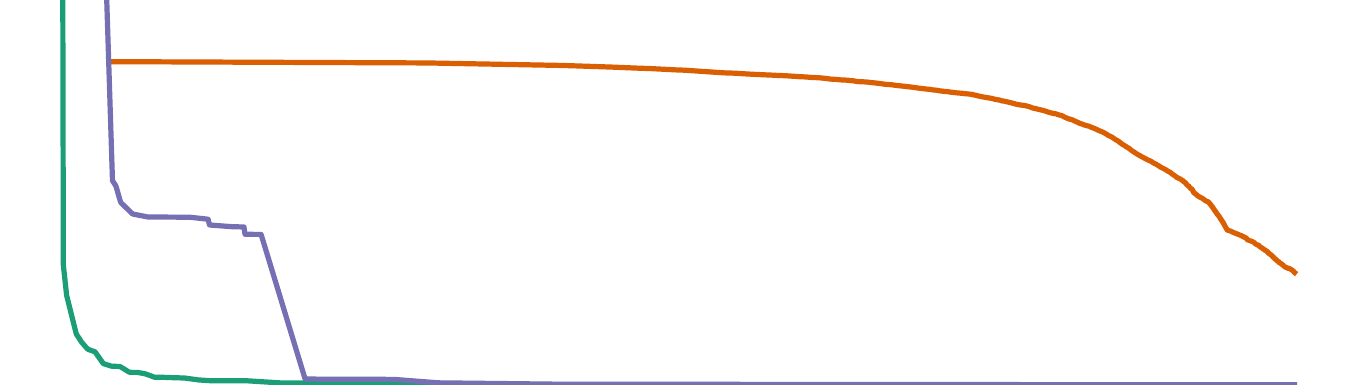}
\begin{tikzpicture}[trim axis group left,trim axis group right]
  \newcommand{\horizontalline}[1]{
    \draw ({1, #1} -| {axis description cs:0, 0}) -| (axis description cs:1, 0);
  }
  \pgfmathsetmacro{\myplotwidth}{(\linewidth - 15mm) / 2}
  \begin{groupplot}[
      group style={
        group size={2 by 2},
        vertical sep=1mm,
        horizontal sep=5mm,
      },
      xmode=log,
      /tikz/mark repeat=1,
      /tikz/mark phase=0,
      width=\myplotwidth,scale only axis,
      height=3cm,
      scaled y ticks=false,
      axis on top,
    ]
    \nextgroupplot[xmin=0.0411112282015524,xmax=54.7549980628163,ymin=381.0,ymax=441.0,xticklabels={},ylabel=bound value (\#),title=\texttt{Texas}] {
      \addplot graphics[xmin=0.0411112282015524, xmax=54.7549980628163, ymin=381.0, ymax=441.0] {figures/evaluation/bounds-gap/bounds-Texas.pdf};
    }
    \nextgroupplot[
      xmin=20.61432111982069,xmax=51419.59945287649,ymin=15096.0,ymax=15161.0,yticklabel pos=right,
      xticklabels={},ylabel=bound value (\#),title=\texttt{USA}
    ] {
      \addlegendimage{area legend,color0,thick,pattern=north east lines,pattern color=.}
      \addlegendentry{Ours}
      \addlegendimage{area legend,color1,thick,pattern=crosshatch dots,pattern color=.}
      \addlegendentry{\texttt{MFN}}
      \addlegendimage{area legend,color2,thick,pattern=north west lines,pattern color=.}
      \addlegendentry{\texttt{MILP}}
      \addplot graphics[xmin=20.61432111982069, xmax=51419.59945287649, ymin=15096.0, ymax=15161.0] {figures/evaluation/bounds-gap/bounds-USA.pdf};
    }
    \nextgroupplot[xmin=0.0411112282015524,xmax=54.7549980628163,ymin=0.0,ymax=0.4,height=1.5cm,xlabel={time (ms)},ylabel=gap] {
      \addplot graphics[xmin=0.0411112282015524, xmax=54.7549980628163, ymin=0.0, ymax=0.4] {figures/evaluation/bounds-gap/gap-Texas.pdf};
    }
    \nextgroupplot[xmin=20.61432111982069,xmax=51419.59945287649,ymin=0.0,ymax=0.4,ymax=0.4,height=1.5cm,xlabel={time (ms)},ylabel=gap,yticklabel pos=right] {
      \addplot graphics[xmin=20.61432111982069, xmax=51419.59945287649, ymin=0.0, ymax=0.4] {figures/evaluation/bounds-gap/gap-USA.pdf};
    }
  \end{groupplot}
\end{tikzpicture}
  \end{center}
  \vspace{-1\baselineskip}
  \caption{Upper and lower bounds on the optimum value on the \texttt{Texas} and \texttt{USA} powersimdata instances with preprocessing by our reduction rules. We give the bounds reported by our solver and by \texttt{MFN}, both with added greedy upper bounds, as well as Gurobi for \texttt{MILP}. Note that the x axis uses a logarithmic scale.}
  \label{fig:exp:bounds and gap}
\end{figure}

\subparagraph{Reduction Rules.}

\begin{figure}
  \centering
  \tikzpicturename{evaluation-running_times}
\tikzpicturedependsonfile{figures/evaluation/running_times_normal.csv}
\tikzpicturedependsonfile{figures/evaluation/speedup_boxes_normal.csv}
\begin{tikzpicture}[trim axis group left]
\newcommand{\horizontalline}[2]{
  \draw ({1, #2} -| {axis description cs:0, 0}) -- node[fill=white,inner sep=0.5mm,font=\scriptsize] {#1} ({1, #2} -| {axis description cs:1, 0});
}
\pgfmathsetmacro{\myplotwidth}{(\linewidth - 5mm) / 2}
\begin{groupplot}[
  group style={group size=2 by 1,horizontal sep=5mm,},
  grid=major,
  /tikz/mark size=1.5pt,
  /tikz/semithick,
  cycle multiindex list={{mark list*}\nextlist{Dark2}\nextlist},
  table/col sep=comma,
  legend columns=6,
  legend style={at={([yshift=2mm]1,1)},anchor=south},
  width=\myplotwidth,scale only axis,
  xlabel shift=-0.3em,
  ymode=log,
  domain=3:12000,
]
  \nextgroupplot[
    xlabel={\# buses}, ylabel={total time (ms)},
    xmode=log, ymin=0.001, ymax=10000,
  ]
  \pgfplotstableread{figures/evaluation/running_times_normal.csv}\loadedtable
  \pgfplotsinvokeforeach{{1ms}{1},{1s}{1000},{1min}{60000},{Timeout (1.5h)}{5400000}}{
    \horizontalline#1
  }
  \foreach \col/\name in {all/all,{non-local}/domination,{local+\texttt{Dom}}/no-necessary,{local+\texttt{NecN}}/no-domination,local/simple,{no reductions}/none}{
    \addplot+[only marks] table[x=n,y=\name] \loadedtable;
    \expandafter\addlegendentry\expandafter{\col};
  }
  \nextgroupplot[
      xtick={0,...,4},
      xticklabels={all,non-local,local+\texttt{Dom},local+\texttt{NecN},local},
      x tick label style={
        anchor=near xticklabel,yshift={-mod(\ticknum,2)*\baselineskip}
      },
      yticklabel pos=right,
      ymajorgrids,
      ylabel={speed-up},
      cycle list name=mark list*,
      /pgfplots/boxplot/.cd,
      draw direction=y,
      every whisker/.style={black},
      every box/.style={black},
      every median/.style={ultra thick},
  ]
    \pgfplotstableread{figures/evaluation/speedup_boxes_normal.csv}\loadedtable
    \pgfplotstabletranspose[colnames from={0}]\stattable{\loadedtable}
    \pgfplotsinvokeforeach{0,1,2,3,4}{
      \pgfmathsetmacro{\myindex}{int(#1+1)}
      \addplot+ [
      color#1,
      boxplot prepared from table={
        table=\stattable,
        row=#1,
        lower whisker=whislo,
        lower quartile=q1,
        median=med,
        upper quartile=q3,
        upper whisker=whishi,
      },
      boxplot prepared={
        box extend=0.35,
        draw position=\plotnumofactualtype,
      },
      area legend,
      fill=none,mark phase=9,
      ] table [y index=\myindex] \loadedtable;
    }
\end{groupplot}

\end{tikzpicture}
  \vspace{-2\baselineskip}

  \subcaptionbox{Median running time on each instance with the different subsets of reduction rules\label{fig:exp:reduction running times}}[.48\linewidth]{~}\hfill
  \subcaptionbox{Aggregated speed-up of each set of reduction rules compared to our algorithm without reduction rules\label{fig:exp:reduction speed up}}[.48\linewidth]{~}

  \caption{Running times  and speed-up of our algorithm with different subsets of the reduction rules on the pandapower instances.}
  \label{fig:exp:running time and sped up}
\end{figure}
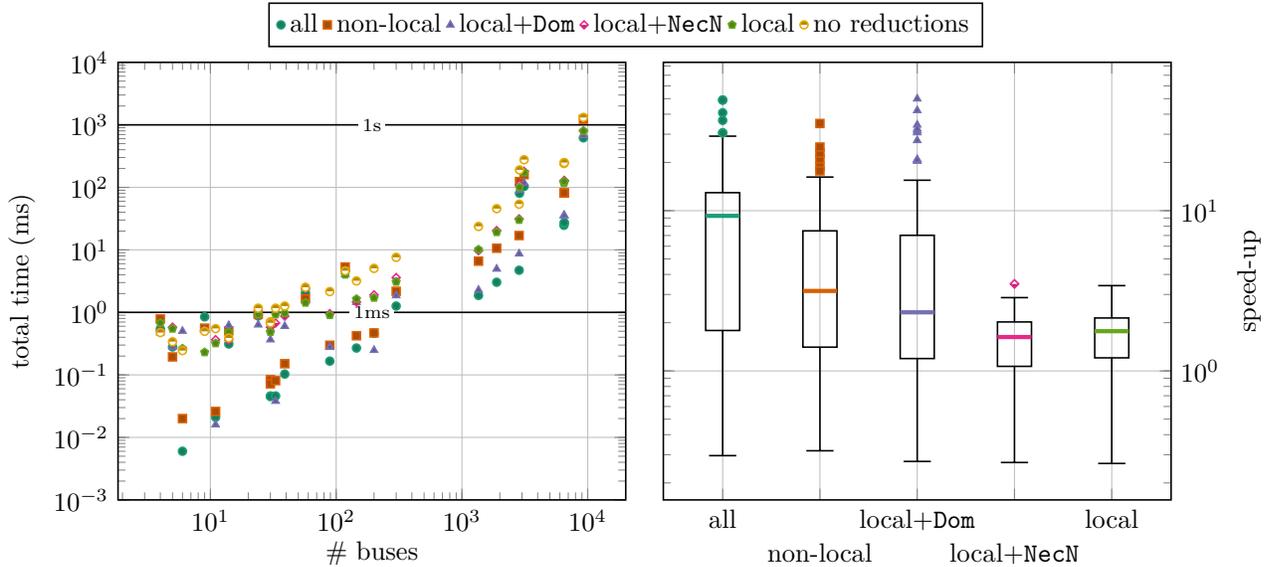

To evaluate the effect of the reduction rules on the performance of our algorithm, we let it run on the pandapower instances with different subsets of reduction rules.
Recall that we have several local reduction rules as well was the two non-local rules \ref{red:domination} and \ref{red:necessary}.
In addition to using all or no reduction rules, we consider the following subsets.
Only local rules, only non-local rules, all local rules together with \ref{red:domination}, and all local rules together with \ref{red:necessary}.

\Cref{fig:exp:reduction running times} shows the median running time for each instance in the different settings.
In most instances, the reductions could decrease the running time by an order of magnitude or more.
Moreover, we can see that in most cases all reduction rules are relevant, i.e., we achieve the lowest run time when using all reduction rules and applying no reduction rules is usually slower than applying any of the rules.

\Cref{fig:exp:reduction speed up} shows the speedup aggregated over all instances of using reduction rules compared to using no reduction rules for our solver.
We can see that the median speedup is roughly one order of magnitude when applying all reduction rules.
The most interesting observation here is that local+\ref{red:necessary} does not give any improvement compared to just local.
In fact, it is slightly slower.
However, when combined with \ref{red:domination}, \ref{red:necessary} gives a significant improvement.


  \section{Conclusion}
\label{sec:conclusion}

We showed that \Pds is $W[P]$-complete.
This closes the gap in the study of its parameterized complexity.
Our reduction uses an auxiliary problem, \Ipds, to simulate arbitrary monotone circuits.

Our second contribution in this paper is a set of new reduction rules for \Pds.
The rules yield partially solved instances of \extension{\Pds} where some vertices are pre-selected for the power dominating set while other are forbidden from being included.
Each rule shrinks the instance by removing vertices or edges, or pre-selects or excludes vertices from being selected.
Our reduction rules can be used as a pre-processing step to significantly enhance the performance of existing solvers.
Our third and last contribution is a new algorithm for solving \Pds based on the implicit hitting set approach.
The core of our algorithm is a new heuristic to find missing sets for the implicit hitting set instances.
We evaluate the effectiveness of our reduction rules and the performance of our algorithm in experiments on a set of practical power grid instances from the literature.
For comparison, we run the same experiments with two different approaches from the literature.
The comparison shows clearly that our new heuristic for finding missing fort neighborhoods outperforms the previous approach.
Our algorithm outperforms the reference solvers by more than one order of magnitude.
Even when combining the other approaches with our reduction rules, our algorithm beats them on most instances.
Furthermore, we can solve large instances of continental scale that could not be solved before.
We found that our algorithm finds lower bounds on the power dominating number more quickly than Gurobi.

A major advantage of our fort heuristic is that it translates easily to other variants of \Pds, as long as it is easy to verify which vertices are observed by a partial solution.
Examples of such variant are the $k$-\PDS where propagation is possible if a vertex has less than $k$ unobserved neighbors or $l$-\textsc{Round \PDS} where the number of propagation steps is limited.
Other variants, such as \textsc{Connected \PDS} are less straightforward.
It might be interesting to see if connectivity can be efficiently enforced in the implicit hitting set model.

Even though our algorithm shows a significant improvement over the state-of-the-art, there is still some potential for further engineering.
Currently, our implementation of the reduction rules is optimized for a single execution as a pre-processing step.
Further optimization might make them more efficient, especially when only few vertices have changed between rule applications.
This might be useful in more accurate heuristics solutions on large instances or for use in a branching algorithm.
Further fast high quality heuristics can provide good upper bounds on the solution size.
Such a heuristic, combined with the lower bound provided by our algorithm, might prove optimality earlier, further reducing the run time.
Also, other hitting set solvers beside Gurobi exist and our algorithm might benefit from using those instead.


  \bibliography{bibliography.bib}
  \bibliographystyle{plainurl}

  \appendix
  \section{Non-Propagation for MILP formulations and Fort neighborhoods}

\subsection{MILP Formulation}
\label{sec:milp formulation}

There exist MILP formulation for \textsc{simple-\Pds}, e.g. \cite{jovanovic2020fixed} and \cite{brimkov2019connected}, but to the best of our knowledge these have not been applied to \extension{\Pds}.
We discuss how these formulations can be modified to solve the extension problem and to accommodate non-propagating vertices.

We use the formulation by Jovanociv and Voss~\cite{jovanovic2020fixed} as a basis for our MILP.
They represent each vertex with two variables $x_v$ and $s_v$.
The binary variable $x_v$ represents whether a vertex is selected.
The third variable $p_{v,w}$ represents a vertex $v$ propagating to another vertex $w$.
To prevent cycles, Jovanovic and Voss count the number of propagation steps on the path from a selected vertex in the variable $s_v$.

Non-propagating vertices can be accounted for by introducing an additional constraint forcing $p_{v,w} = 0$ if $v$ is non-propagating.
If a vertex $v$ is excluded, we add a constraint $x_v = 0$ and if $v$ is selected, we add a constraint $x_v = 1$.

While Jovanovic and Voss require $s_v$ to be integral, we relax this constraint and allow continuous values instead.
The steps only require a minimum offset but not integrality.
There are also several other redundant constraints in the model.
We did not remove all of them as we found in preliminary experiments that they improve the solver performance.
We did remove one constraint, namely $p_{v,w} + p_{w,v} \leq 1$.
This constraint follows from constraint~\labelcref{eq:milp:step} and seems to lead to a significant decrease in performance.
For completeness we state the resulting model with all our modifications.

\begin{align}
  \label{eq:milp}
  \text{Minimize } &\sum_{v\in V} x_v\\
  \text{s.t. } & s_v \leq x_t + M (1 - x_t) & \forall v \in V, t \in N[v]\label{eq:milp:domination}\\ 
  &s_v \leq M (x_v + \sum_{w \in N(v)} (x_w + p_{w,v})) &\forall v \in V \label{eq:milp:observation rules}\\
  &\sum_{w\in N(v)} p_{w,v} \leq 1 & \forall v \in V \label{eq:milp:observing arc in}\\
  &\sum_{w\in N(v)} p_{v,w} \leq 1 & \forall v \in V \label{eq:milp:observing arc out}\\
  &s_v \geq s_t + 1 - M (1 - p_{w,v}) &\begin{aligned} \forall &v \in V, w \in N(v),\\ &t \in N[w] \setminus \set{v}\end{aligned} \label{eq:milp:step}\\
  &1 \leq s_v \leq \abs{V}  \forall v \in V\label{eq:milp:observation guarantee}\\
  &x_v = 1 & \forall v \in X \label{eq:milp:selected} \\
  &x_v = 0 & \forall v \in Y \label{eq:milp:excluded} \\
  &p_{v,w} = 0 & \forall v \in Z, w \in N(v) \label{eq:milp:non propagating}\\
  &x_v,p_{w,v} \in \set{0,1} \label{eq:milp:binary}
\end{align}

\begin{lemma}
  Let $G$ be a \extension{\Pds}-instance.
  The MIP formulation has a solution with objective value $k$ if and only $G$ has a solution of size $k$.
\end{lemma}

\begin{proof}
  Let $G=(V, E)$ be a \extension{\Pds}-instance and let its MIP be defined as above.
  \begin{claim}[If $G$ has a solution of size $k$, there is a satisfying assignment $(x, s)$]
    Let $S$ be a minimum power dominating set of $G$ and let $r_1,\dots,r_\ell$ be a sequence of rule applications that observes $G$.
    Without loss of generality, we assume that all applications of the domination rule occur before the first application of the propagation rule.
    Every vertex is observed either by the domination rule or by the observation rule.
    If $v$ is observed by the domination rule we set $x_v = s_v = 1$, otherwise we set $x_v = 0$ and thus constraint~\labelcref{eq:milp:domination,eq:milp:selected,eq:milp:excluded,eq:milp:observation guarantee} are satisfied.
    Only if a vertex $u$ observes another vertex $v$ by the propagation rule, we set $p_{u,v} = 1$ and $s_v = 1 + \max\set{s_w \mid w \in N[u] \setminus \set{v}}$ which satisfies constraint~\labelcref{eq:milp:observation rules}.
    By default we set $p_{u,v}=0$, in particular for all non-propagating vertices, satifying constraints~\labelcref{eq:milp:step,eq:milp:observing arc in,eq:milp:observing arc out}.
    We only assign 1 or 0 to the variables, so constraint~\labelcref{eq:milp:binary} is satisfied.
  \end{claim}
  \begin{claim}[If the MIP has a solution of weight $k$, then $G$ has a solution of size $k$]
    We construct a candidate solution $S = \set{v \mid x_v = 1}$.
    Let $V = v_1,\dots,v_{\abs{V}}$ be sorted by the value of $s_v$ in ascending order.
    We show inductively that $S$ is a solution of $G$ by iterating the vertices sorted by their value of $s_v$ in ascending order.
    In each step $i$ we assume that there is a sequence of observation rules that observes all vertices up to $v_{i-1}$ and we extend this sequence to observe $v_i$.

    First, we apply the domination rule to all vertices in $S$, i.e. with $x_v = 1$.
    This observes $N[S]$.
    By constraints~\labelcref{eq:milp:domination,eq:milp:observation guarantee} all vertices in $N[S]$ must have $s_v = 1$.
    Note that by constraint~\labelcref{eq:milp:observation guarantee} there is no vertex with $s_v < 1$.

    Now let $u$ be a vertex with $s_v > 1$.
    By constraint~\labelcref{eq:milp:domination} we know that $N[u] \cap S$ is empty and thus constraint~\labelcref{eq:milp:observation rules} implies that there is at least one other vertex $v \in N(u)$ with $p_{v,u} = 1$.
    Constraint~\labelcref{eq:milp:observing arc in} ensures that there is only one such $v$.
    Then constraint~\labelcref{eq:milp:step} states that $s_u > s_w$ for all vertices $w \neq u$ in the closed neighborhood $N[v] \setminus \set{u}$.
    By our induction hypothesis we thus know that all vertices in $N[v] \setminus \set{u}$ are observed, i.e. $v$ is observed and has only one unobserved neighbor, $u$.
    We can thus apply the propagation rule and $u$ becomes observed.
  \end{claim}
\end{proof}

In the above formulation we account for propagation and selected and excluded vertices by introducing additional constraints.
Instead of introducing additional constraints, we can modify the existing constraints to account for these additional properties.
The resulting model contains fewer variables and constraints from the start.

\subsection{Fort Neighborhoods}

Both previous approaches using fort neighborhoods search for fort neighborhoods of minimum size~\cite{bozeman2018RestrictedPD,smith2020OptimalSP}.
We adapt their ILP formulations for the two subproblems of finding a hitting set and finding a violated fort to \extension{\Pds}.
Let $G=(V, E)$ be a \extension{\Pds}-instance with selected vertices $X$, excluded vertices $Y$ and non-propagating vertices $Z$.
We introduce a variable $s_v$ for each vertex, representing whether that vertex is in the power dominating set.
The following model finds a minimum hitting set of $M$ that contains all vertices in $X$ and no vertices in $Y$.
\begin{align}
  \label{eq:hitting set}
  \text{Minimize } & \sum_{v \in V} s_v\\
  \text{s.t. } &\sum_{v \in F} s_v \geq 1 & \forall F \in M\\
  &s_v = 1 &\forall v \in X\\
  &s_v = 0 &\forall v \in Y\\
  &s_v \in \set{0, 1}
\end{align}

To find a violated fort neighborhood of minimum size, we adapt the method introduced by Smith and Hicks~\cite{smith2020OptimalSP}.
Let $S$ be a set of selected vertices and let $R$ be the set of vertices observed by exhaustive application of the observation rules.
For each vertex we introduce two variables $x_v$ and $y_v$ which represent whether $v$ is in the fort or in the neighborhood of the fort, respectively.
\begin{align}
  \label{eq:fort neighborhood}
  \text{Minimize } & \sum_{v \in V} y_v\\
  \text{s.t. } & \sum_{v \in V} x_v \geq 1\\
  & x_v = 0 & \forall v \in R\\
  & x_v + \sum_{w \in N(v)} \geq x_u & \forall u \in V \,\forall v \in N(u)\\
  & y_v \geq x_w &\forall v \in V \,\forall w \in N[v]\\
  & x_v \in \set{0, 1}
\end{align}
Note that this model has a feasible solution if and only if $R \neq V$.

\section{Reduction Rules (Missing Proofs From Section~\ref{sec:reduction rules})}
\label{sec:reduct-rules-appendix}

\reductionDegIa*
\reductionDegIb*
\begin{lemma}
  \Cref{red:leaf 1,red:leaf 2} are safe.
\end{lemma}
\begin{proof}
  Given a graph $G = (V, E)$ and sets $X, Y$ as above, let $v \in V$ be a leaf with neighbor $w \in V$.
  Let $S$ be a minimum solution.
  If $v \in S$ and $w \notin Y$, we can construct an equivalent solution $S' = (S \setminus \set{v}) \cup \set{w}$.
  Therefore, \labelcref{red:leaf 1} is safe.
  For the second rule we distinguish two cases.

  First, if $w$ is propagating and $v$ excluded, $w$ cannot propagate to any vertex but $v$.
  Because a power dominating set requires all vertices to be observed, this is always the case in a valid solution.
  It is thus safe to remove $v$ and set $w$ to non propagating.

  In the second case $w$ is already non-propagating.
  In this case $v$ can only be observed from $w$.
  Thus $w$ must be selected.
\end{proof}

\reductionTri*

\begin{lemma}
  \Cref{red:tri} is safe.
\end{lemma}
\begin{proof}
  At least one of $x$, $y$ and $z$ must be selected, otherwise $x$ and $y$ remain unobserved.
  All vertices that would become observed by selecting $x$ or $y$ also become observed when selecting $z$.
  The two vertices $x$ and $y$ become observed when selecting $z$ and have no remaining unobserved neighbors or neighbors that could propagate.
  They can thus safely be removed.
\end{proof}

\reductionDegIIa*

\reductionDegIIb*

\reductionDegIIc*

\reductionOnlyN*

\begin{lemma}
  \Cref{red:tri,red:path 5} are safe.
\end{lemma}
\begin{proof}
    The only way $v$ can be observed is by selecting $y$.
\end{proof}
\begin{lemma}
  \Cref{red:path 1,red:path 2,red:path 3} are safe. \cref{,,red:tri,,red:path 5} are safe.
\end{lemma}
\begin{proof}
  These rules are based on the observation that vertices of degree propagate observation whenever they become observed by one of their neighbors.
  \begin{enumerate}[ncases]
    \item[\labelcref{red:path 1}] One can always select the non-excluded neighbor instead of $v$ and observe the other neighbor by the propagation rule.
    \item[\labelcref{red:path 2}] Without loss of generality, assume that $x$ has degree two.
      If $v$ becomes observed by one of its neighbors, the other neighbor becomes observed by the propagation rule.
      This is still the case after applying the reduction.
      Further, if $x$ becomes observed in the reduced instance, it is either observed from $y$ and thus $v$ is observed in the original instance, or $x$ is observed from some other vertex in which case $v$ and $y$ become observed by the propagation rule from $x$.
    \item[\labelcref{red:path 3}] The observed vertex $v$ has two unobserved neighbors of degree two, $x$ and $y$ which in turn have neighbors $z_x$ and $z_y$.
      By the definition of the rule, $v$, $x$, $y$, $z_x$ and $z_y$ are distinct vertices.
      This structure is a special case of the booster gadget from \cref{lem:booster gadget} where both endpoints of the booster edge have degree two.
      Thus, when $x$ becomes observed, $y$ becomes observed and vice versa.
      Both vertices have at most one unobserved neighbor, so when they become observed, that neighbor becomes observed, too, by the propagation rule.
      To their neighbors $z_x$ and $z_y$, $x$ and $y$ thus behave like a single vertex of degree two.
  \end{enumerate}
\end{proof}

\reductionObsNP*
\begin{lemma}
  \Cref{red:observed non-zi} is safe.
\end{lemma}
\begin{proof}
  Let $v$ be an observed inactive non-zero-injection vertex as above.
  Such a vertex can never propagate and does not count toward the unobserved neighbors of any other vertex.
  The observation rules can thus not be applied to $v$ and application of the propagation rule to any of the neighbors of $v$ does not depend on $v$.
  It is thus safe to remove $v$.
\end{proof}

\reductionObsE*

\begin{lemma}
  \Cref{red:observed edge} is safe.
\end{lemma}
\begin{proof}
  After application of the reduction, both vertices have a selected neighbor and can thus be observed by the domination rule, so they remain observed.
  The their number of unobserved neighbors does not change and thus the propagation rule can be applied to them as before.
\end{proof}

\reductionIsol*
\begin{lemma}
  \Cref{red:isolated} is safe.
\end{lemma}
\begin{proof}
  Isolated vertices cannot be observed by any other node so they must be selected.
\end{proof}

\reductionDom*

\begin{lemma}
  \Cref{red:domination} is safe.
\end{lemma}
\begin{proof}
  Because $N[w]$ is contained in the observation neighborhood of $v$, $w$ can never observe any vertices that are not observed when selecting $v$.
  Thus a power dominating set containing $w$ can never be smaller than one containing $v$ and $v$ can safely be excluded.
\end{proof}

\reductionNecN*

\begin{lemma}
  \Cref{red:necessary} is safe.
\end{lemma}
\begin{proof}
  Selecting more vertices cannot make fewer vertices observed.
  Thus, if the observation neighborhood of all blank vertices except $v$ does not contain the whole graph, there exists no power dominating set that does not contain $v$.
\end{proof}

%


\end{document}